\documentclass[%
reprint,
superscriptaddress,
 amsmath,amssymb,
 aps,
]{revtex4-1}
\usepackage{float}
\usepackage{subcaption}  

\usepackage{mathtools}
\usepackage{dcolumn}
\usepackage{bm}
\usepackage{physics}

\newcounter{conditioncount} 

\usepackage{booktabs}
\usepackage{colortbl}
\usepackage{diagbox}
\usepackage{array}
\usepackage{amsmath}
\usepackage{amssymb}
\usepackage{amsthm}
\usepackage{graphicx}
\usepackage{algorithmic}
\usepackage{mathrsfs}
\usepackage{braket}
\usepackage{xcolor}
\usepackage{tikz}
\usepackage[roman]{complexity}
\usepackage[caption=false]{subfig}

\newtheorem{theorem}{Theorem}

\usepackage{pdfpages}
\usepackage{pgffor}
\makeatletter
\AtBeginDocument{\let\LS@rot\@undefined}
\makeatother

\begin{document}

\title{The Stabilizer Bootstrap of Quantum Machine Learning with up to 10000 qubits}

\author{Yuqing Li}
\email{yli@pitt.edu}
\affiliation{Department of Computer Science, University of Pittsburgh, Pittsburgh, PA 15260, USA}

\author{Jinglei Cheng}
\email{jic373@pitt.edu}
\affiliation{Department of Computer Science, University of Pittsburgh, Pittsburgh, PA 15260, USA}

\author{Xulong Tang}
\email{xulongtang@pitt.edu}
\affiliation{Department of Computer Science, University of Pittsburgh, Pittsburgh, PA 15260, USA}

\author{Youtao Zhang}
\email{youtao@pitt.edu}
\affiliation{Department of Computer Science, University of Pittsburgh, Pittsburgh, PA 15260, USA}

\author{Frederic T. Chong}
\email{chong@cs.uchicago.edu}
\affiliation{Department of Computer Science, University of Chicago, Chicago, IL 60637}

\author{Junyu Liu}
\email{junyuliu@pitt.edu}
\affiliation{Department of Computer Science, University of Pittsburgh, Pittsburgh, PA 15260, USA}

\date{\today}

\begin{abstract}

Quantum machine learning is considered one of the flagship applications of quantum computers, where variational quantum circuits could be the leading paradigm both in the near-term quantum devices and the early fault-tolerant quantum computers. However, it is not clear how to identify the regime of quantum advantages from these circuits, and there is no explicit theory to guide the practical design of variational ans\"{a}tze to achieve better performance. We address these challenges with the \texttt{stabilizer bootstrap}, a method that uses stabilizer-based techniques to optimize quantum neural networks before their quantum execution, together with theoretical proofs and high-performance computing with 10000 qubits or random datasets up to 1000 data. We find that, in a general setup of variational ans\"{a}tze, the possibility of improvements from the stabilizer bootstrap depends on the structure of the observables and the size of the datasets. The results reveal that configurations exhibit two distinct behaviors: some maintain a constant probability of circuit improvement, while others show an exponential decay in improvement probability as qubit numbers increase. These patterns are termed \emph{strong stabilizer enhancement} and \emph{weak stabilizer enhancement}, respectively, with most situations falling in between. Our work seamlessly bridges techniques from fault-tolerant quantum computing with applications of variational quantum algorithms. Not only does it offer practical insights for designing variational circuits tailored to large-scale machine learning challenges, but it also maps out a clear trajectory for defining the boundaries of feasible and practical quantum advantages.

\end{abstract}

\maketitle

\section{Introduction}

Quantum machine learning (QML) represents a powerful approach with demonstrated applications across multiple scientific domains, from protein folding and drug discovery in biology \cite{wang2024comprehensive,batra2021quantum,cao2018potential} to molecular structure optimization in chemistry \cite{peruzzo2014variational}. The performance of QML algorithms depends heavily on initial conditions, which is a characteristic shared with classical machine learning. 
Research shows that optimal initialization of training circuits can reduce convergence time and improve final results \cite{glorot2010understanding,he2015delving}. 
The classical bootstrap~\cite{el2014solving}, a process for searching and optimizing certain parameters under constraints, serves as a crucial step in developing effective variational quantum algorithm initial parameters \cite{ravi2022cafqa}.

In general, recent research has established several parameter initialization approaches for QML, including tensor network methods \cite{rudolph2023synergistic}, Gaussian initialization techniques \cite{zhang2024escaping}, matrix product state optimization \cite{dborin2021matrix}, and deep neural network integration \cite{shi2024avoiding}. These initialization methods have advanced QML performance through reduced training epochs, faster loss function convergence, and improved final accuracy. The prior work \cite{ravi2022cafqa} introduced CAFQA, a novel approach that searches Clifford space with Bayesian optimization to identify optimal quantum states for Variational Quantum Eigensolver (VQE) tasks, naturally leveraging concepts from fault-tolerant quantum computing towards practical applications. In various experiments, CAFQA surpasses the traditional chemical approach of finding a suitable computational basis state known as Hartree-Fock (HF) \cite{hartree1935self} initialization and even could show $2.5 \times$ faster convergence than HF for small molecules. While these advances show important progress, current research remains limited to small-scale systems and has no discussion of the dependence on involved dataset. In fact, since the optimization uses Clifford circuits that can be simulated classically due to the famous Gottesman-Knill theorem \cite{gottesman1998heisenberg}, one could in principle consider extremely large-scale simulations and fruitful data sets beyond the capability of state-vector simulators.

In our work, we address such challenges by utilizing high-performance computing to advance QML simulations to the next level, a process we refer to as the \texttt{stabilizer bootstrap} \footnote{The term \emph{stabilizer bootstrapping}, similar to our term, is used in \cite{chen2024stabilizer} but only for quantum tomography. }. To set up our problem, we consider a general class of variational circuits with varying measurement observables and entanglement structures. We employ stabilizer circuits consisting of layers of $R_y$ gates and CNOT gates to assess algorithm performance. Specifically, we investigate different CNOT layer structures, including linear and reverse-linear entanglement configurations. These circuits correspond to continuous variational quantum circuits with special variational angles ${0, \frac{\pi}{2}, \pi, -\frac{\pi}{2}}$. We study how much improvement can be achieved for a given trivial initial state by selecting this special set of angles, a concept we term \emph{stabilizer enhancement}. The \texttt{stabilizer bootstrap} process involves sampling all possible combinations of ${0, \frac{\pi}{2}, \pi, -\frac{\pi}{2}}$ to find the optimal solution that minimizes a given loss function derived from quantum measurements of these circuits.

Regarding operator designs in our QML problems, we evaluate the algorithm's performance under observables composed exclusively of Pauli-$X$ and Pauli-$Z$ operators, varying the proportion of Pauli-$X$ operators in the observables. The most extreme cases involve either all $X$ or all $Z$, while intermediate mixtures are quantified by the ratio $r$, representing the proportion of $X$ operators in the observable string. The purpose of our \texttt{stabilizer bootstrap} program is to scale high-performance computing to large numbers of qubits and large datasets, achieving simulations with up to 10000 qubits and datasets of size up to 1000. Our optimization process for the \texttt{stabilizer bootstrap} in QML follows two critical phases. The first phase involves point sampling, while the second phase focuses on parameter optimization. This optimization phase comprises two key components: a random forest model and a greedy acquisition function, both designed for discrete parameter spaces.

Our experimental results highlight that the sampling phase is particularly crucial—without non-trivial sample points, the parameter optimization cannot proceed effectively. Intuitively, one might expect the potential for such enhancement to decay exponentially with the number of qubits, given the exponentially large parameter space. This reinforces the understanding that quantum circuits cannot be enhanced purely classically when approaching the regime of quantum advantage.
However, we find that it is not always true. There do exist cases where one can always enhance the original trivial state with constant probability, which we call it \emph{strong stabilizer enhancement}. We can also call the exponential decay with maximal exponent as \emph{weak stabilizer enhancement}. The most extreme cases (both strong and weak) might happen at $r=0$ or $r=1$, and most cases are between those two $0<r<1$, which, in a sense, define a practical version of classical simulatability and indicate a boundary towards possible quantum advantages. The most extreme strong and weak cases can be theoretically proven. Although we expect that there might be a more elegant proof using group theory and stabilizers (where we leave it for future research), we show it directly with proof by induction. Cases in the middle could be justified using our bootstrap program with large-scale simulations, where we can actually measure how \emph{exponential} the dependence could be by the \emph{critical exponent} $\nu$, where the enhancement decays as $1/n^\nu$ with the number of qubits $n$. The exponential decay cases correspond to $\nu \sim n/\log n$ (where in physics language it is called a \emph{constant gap}), while the strong enhancement correspond to $\nu =0$. Moreover, we study how large the enhancement will be depending on the size of datasets, where we encode data randomly in the Clifford circuit and use a joint loss function. Numerical evidence shows that it is likely to be exponentially decaying with the data size.   

The overall logic of our paper is summarized in Figure \ref{fig:overview}, where technical details are given in the Appendix, including background reviews, theoretical proofs, and experimental details.  

\begin{figure}
    \centering
    \includegraphics[width=\linewidth]{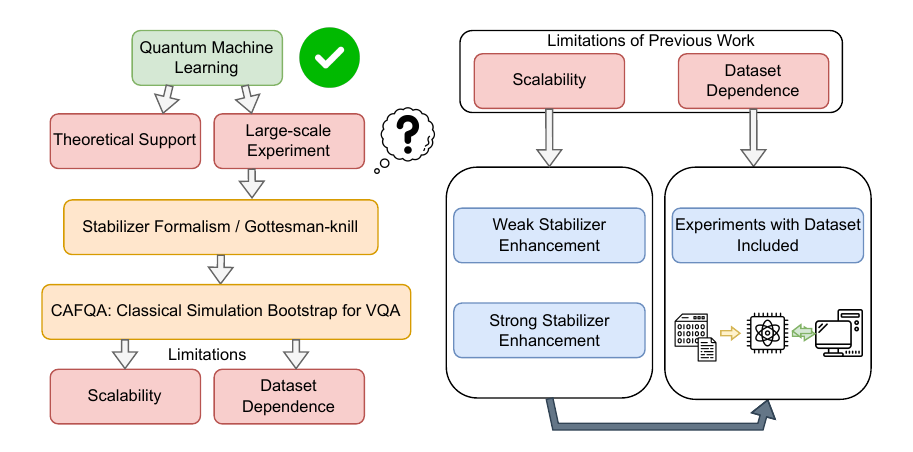}
    \caption{Overview of our paper. Firstly, we present results for single-layer stabilizer ans\"{a}tze under basic conditions, which we denote as \textit{strong stabilizer enhancement} and \textit{weak stabilizer enhancement}. Finally, we present our experimental results for the \texttt{stabilizer bootstrap} with datasets, analyze the impact of dataset size, and demonstrate outcomes for systems with 10000 qubits or 1000 data sizes maximally.}
    \label{fig:overview}
\end{figure}
\section{Results}
\label{results}
In this section, we present our work's core results, including key conclusions and primary experimental results. 

The stabilizer bootstrap employs Bayesian optimization to efficiently explore the space of Clifford operations and optimize parameters. 
To provide a clearer understanding of this approach, we first outline the fundamental principles of Bayesian optimization and then highlight the importance of sampling. 
Bayesian optimization consists of two main phases: sampling and optimization.
The sampling phase establishes the initial set of points for model training. 
The optimization phase then proceeds through three sequential steps based on sampled points: (1) Point Selection: Choose the next points to sample based on the acquisition function; (2) Evaluation: Calculate the objective function values for these points. (3) Model Update: Refine the surrogate model with the newly evaluated points. 
According to optimization theory, if we fail to sample enough useful or nontrivial points, subsequent optimization based on these samples will become challenging.
Consider an extreme case: if all sampled points are trivial (e.g., zero in this experiment), the acquisition function will continue selecting the next points far away from current points randomly until it identifies some promising ones. 
The detailed reason behind this can be found in Appendix \ref{supp:A}. 
Besides, since sampling is easier to accelerate by CPU parallelization, we prefer to allocate more computational resources to sampling rather than optimization to reduce the overall runtime.
Therefore, in our work, we pay much attention to improve sampling efficiency. 
How can we quantify sampling efficiency? 
Our work employs the stabilizer for enhancement, where the outcomes are restricted to the set \(\{0, 1, -1\}\). 
Through extensive experiments, we observed that in large-scale qubit systems, sampling outcomes of 1 or $-1$ become increasingly rare, with the results most likely being 0. 
Therefore, we define the probability of nontrivial sampling—specifically, the probability of obtaining a result of 1 during sampling—as a key metric for sampling efficiency (Due to symmetry, we assume the probabilities of obtaining 1 and $-1$ are equal, and therefore, we focus solely on the probability of obtaining 1 for our experiments). 

We begin our discussion about sampling efficiency from the easiest single-layer stabilizer ans\"{a}tze, with only $R_y$ gates. 
We consider two variations of such Clifford ans\"{a}tz with different entanglement structures as shown in FIG.\ref{fig:clifford_ansatz}. 
Note that our discussion here is not yet related to the database. 
All our experiments are conducted starting from the initial quantum state $| 0 \ldots 00 \rangle$.

\begin{figure}[H]
	\centering
	\includegraphics[width=\linewidth]{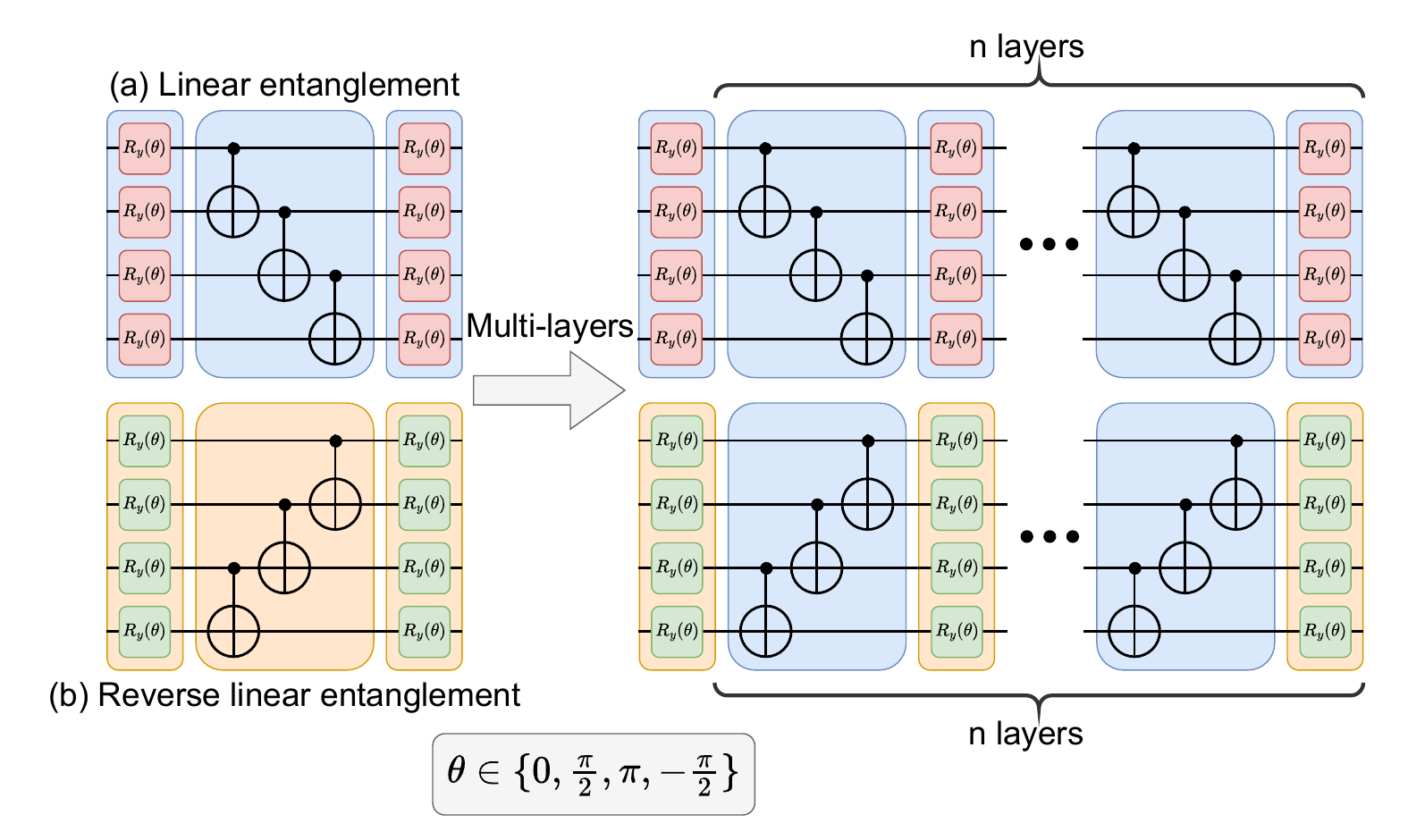}
	\caption{Clifford ans\"{a}tze. For simplification, We refer to the upper (a) entanglement structure as linear entanglement and the below (b) entanglement structure as reverse linear entanglement.}
	\label{fig:clifford_ansatz}
\end{figure}

In accordance with the observables and entanglement structures depicted in FIG. \ref{fig:clifford_ansatz}, we consider four distinct conditions: (1) measurement observables is $X..X$, and the entanglement structure is linear entanglement; (2) measurement observables is $Z..Z$, and the entanglement structure is linear entanglement; (3) measurement observables is $X..X$, and the entanglement structure is reverse linear entanglement; (4) measurement observables is $Z..Z$, and the entanglement structure is reverse linear entanglement. 
We demonstrate the probability of measuring 1 or $-1$ under these four conditions in the TABLE \ref{tab:prob}. 
The corresponding theorems could be found in Appendix \ref{supp:B}.

\begin{figure}[h!]
    \centering
    \includegraphics[width=0.9\linewidth]{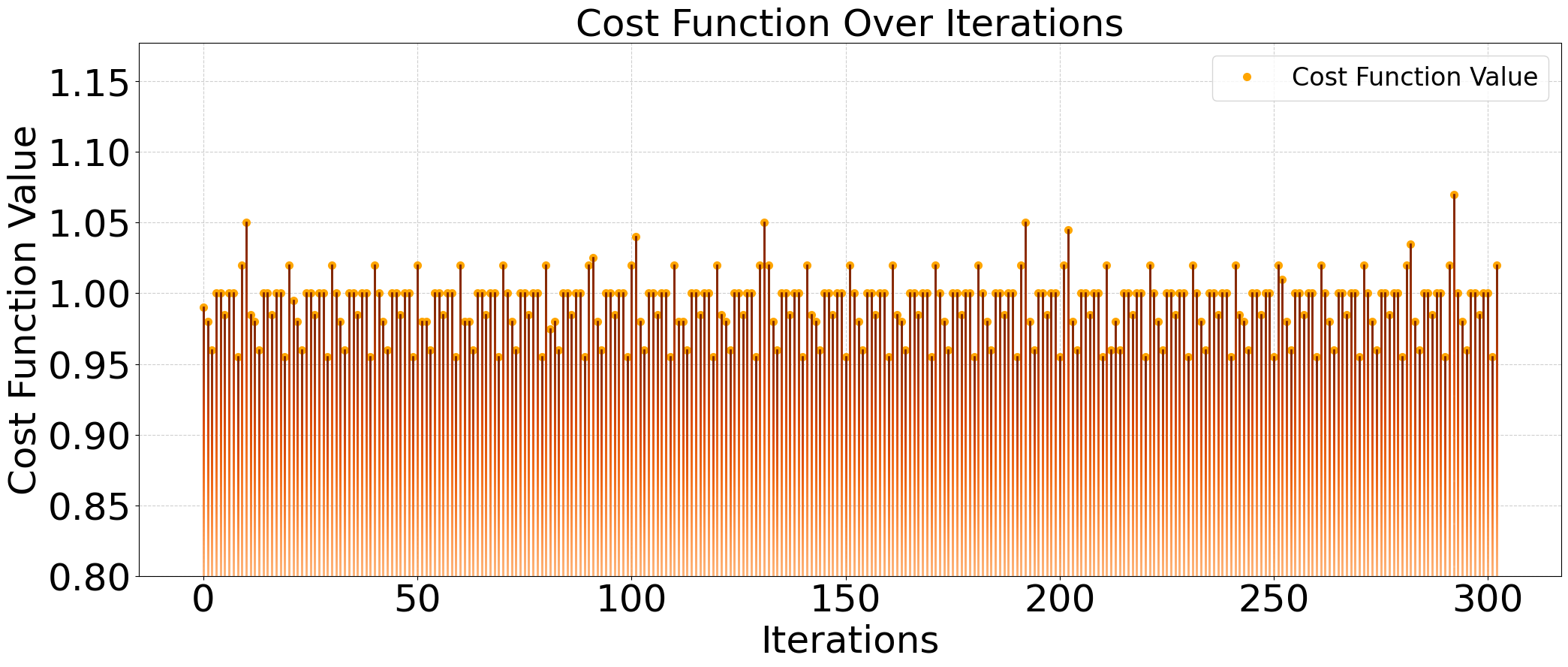}
\caption{Result for 10000 qubits with 200 data.}
    \label{fig:1wqubits}
\end{figure}

To conduct experiments for this paper, we generate binary classification datasets with the \texttt{sklearn.datasets.make\_classification } function up to 1000 data. We construct the stabilizer with the \texttt{Stim} library and employ \texttt{hypermapper} for Bayesian optimization. Our loss function is the mean squared error (MSE). Here, we demonstrate the experiment up to 10000 qubits using a dataset with 200 data in FIG. \ref{fig:1wqubits}, representing our hardware's maximum capacity. All experiments were performed on 8 computing nodes running Debian GNU/Linux 11, each equipped with an AMD EPYC processor featuring 16 cores (32 logical processors) and 251 GB of RAM.

\renewcommand{\arraystretch}{2.5} 

\begin{table*}[ht]
\caption{\textbf{Probability of measuring $1$ or $- 1$}}
\centering
\begin{tabular*}{0.65\textwidth}{@{\extracolsep{\fill}}ccc@{}}
\hline
& \multicolumn{2}{c}{\raisebox{0.9em}{\textbf{\underline{\quad \qquad Observables \qquad \quad}}}} \\ [-2.5em]
\textbf{Entanglement structures} & $\displaystyle X \ldots X$ & $\displaystyle Z \ldots Z$ \\ [-0.5em]
\hline
Linear entanglement structure in FIG. \ref{fig:clifford_ansatz}
& \refstepcounter{conditioncount}\label{condition1} $\displaystyle \frac{1}{2^{\lceil n/2 + 1 \rceil}}$
& \refstepcounter{conditioncount}\label{condition2} $\displaystyle \frac{1}{4}$ \\

Reverse linear entanglement structure in FIG. \ref{fig:clifford_ansatz}
& \refstepcounter{conditioncount}\label{condition3} $\displaystyle \frac{1}{4}$ 
& \refstepcounter{conditioncount}\label{condition4} $\displaystyle \frac{1}{2^{\lceil n/2 + 1 \rceil}}$ \\
\hline
\end{tabular*}
\label{tab:prob}
\end{table*}

Under condition (\ref{condition2}) and (\ref{condition3}), the probability is always constant, allowing us to extend our algorithms to large-scale qubit systems easily. 
Therefore, we define this phenomenon as strong stabilizer enhancement, as the stabilizer can readily find good initial points within this space. 
We argue that under strong stabilizer enhancement, the quantum advantage may be limited or even vanish due to the high efficiency of classical simulations.

Under condition (\ref{condition1}) and (\ref{condition4}), regardless of parity, we treat it as exponential decay and refer to it as weak stabilizer enhancement since the stabilizer can efficiently identify a good initial state when the number of qubits is large.
In the Appendix \ref{supp:B}, we provide a rigorous proof of the condition (\ref{condition2}) and condition (\ref{condition4}) using mathematical induction. 
Due to symmetry, if we simultaneously modify the entanglement structure and replace the \( Z \) operator with the \( X \) operator, the probability remains unchanged. 
Consequently, we could derive the result for condition (\ref{condition1}) and condition (\ref{condition3}).

For simplicity, we primarily focus on the reverse linear entanglement structure in FIG.\ref{fig:clifford_ansatz} in the following section. 
We have proved the probability of nontrivial sampling for commonly used $Z$-string observables and $X$-string observables. 
We are also interested in more general conditions, and here, we would extend our discussion to more general observables. 
The experimental results for the observables consisting of $X$ operators and $Z$ operators and Clifford stabilizer are shown in FIG.\ref{fig:clifford_ansatz}.
We note that the probability of nontrivial sampling is always between $0$ and $\frac{1}{4}$. 
It is a conclusion drawn from various experiments, and detailed experiment results can be found in Appendix \ref{supp:C}.

\begin{figure}[H]
    \centering
    \includegraphics[width=0.7\linewidth]{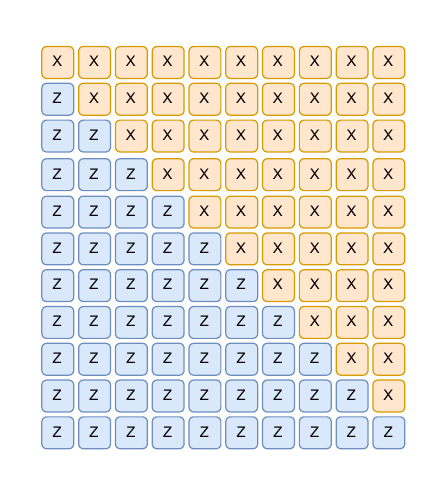}
    \caption{Domain-wall observables.}
    \label{fig:domian-wall observables}
\end{figure}

We are particularly interested in the relationship between this probability and the structure of the observables. 
Specifically, we would like to understand how this probability decays with the number of qubits under different observable structures.
To investigate this, we conduct experiments to determine the decay rate under the domain-wall observable structure, $ZZ...XXX$, with different proportions of $X$ operators in the observable in FIG. \ref{fig:domian-wall observables}. 
Considering the parity issue shown in condition (\ref{condition4}), we analyze two cases separately: one for scenarios where the number of qubits is odd and the other for when it is even. 
The results are included in FIG.\ref{fig:even_prob} and FIG.\ref{fig:odd_prob}.

We hypothesize there is a polynomial decay between exponential decay and constant. To deepen our understanding of this phenomenon, we try to use the following function to approximate the decay carves in all observables:
$$p (r, n) = \frac{1}{4 n^\nu } \label{equ:v},$$
where $r$ represents the proportion of $X$ operators in the observable and $n$ is the number of qubits. 
We argue that with different $r$, $\nu$ is also different. For example, under conditions of strong stabilizer enhancement and $n$ is even, we have:
$$\frac{1}{4 n^\nu } = \frac{1}{4} \rightarrow \nu = 0,\label{equ:0}$$
and under conditions of weak stabilizer enhancement, we have:
$$\frac{1}{4 n^\nu } = \frac{1}{2^{\frac{n}{2} + 1}} \rightarrow \nu = \frac{\left(
\frac{n}{2} - 1 \right) \log 2}{\log n}.\label{equ:logn}$$
Also, we demonstrate detailed exponents in FIG. \ref{fig:exponents}. 
Some data points extend beyond the red line range, primarily as a result of experimental errors. 
Apart from this, the decay rates of other cases consistently fall between the exponential decay and the constant value.

\begin{figure}[H]
    \centering
    \begin{subfigure}[b]{0.49\linewidth}
        \centering
        \includegraphics[width=\linewidth]{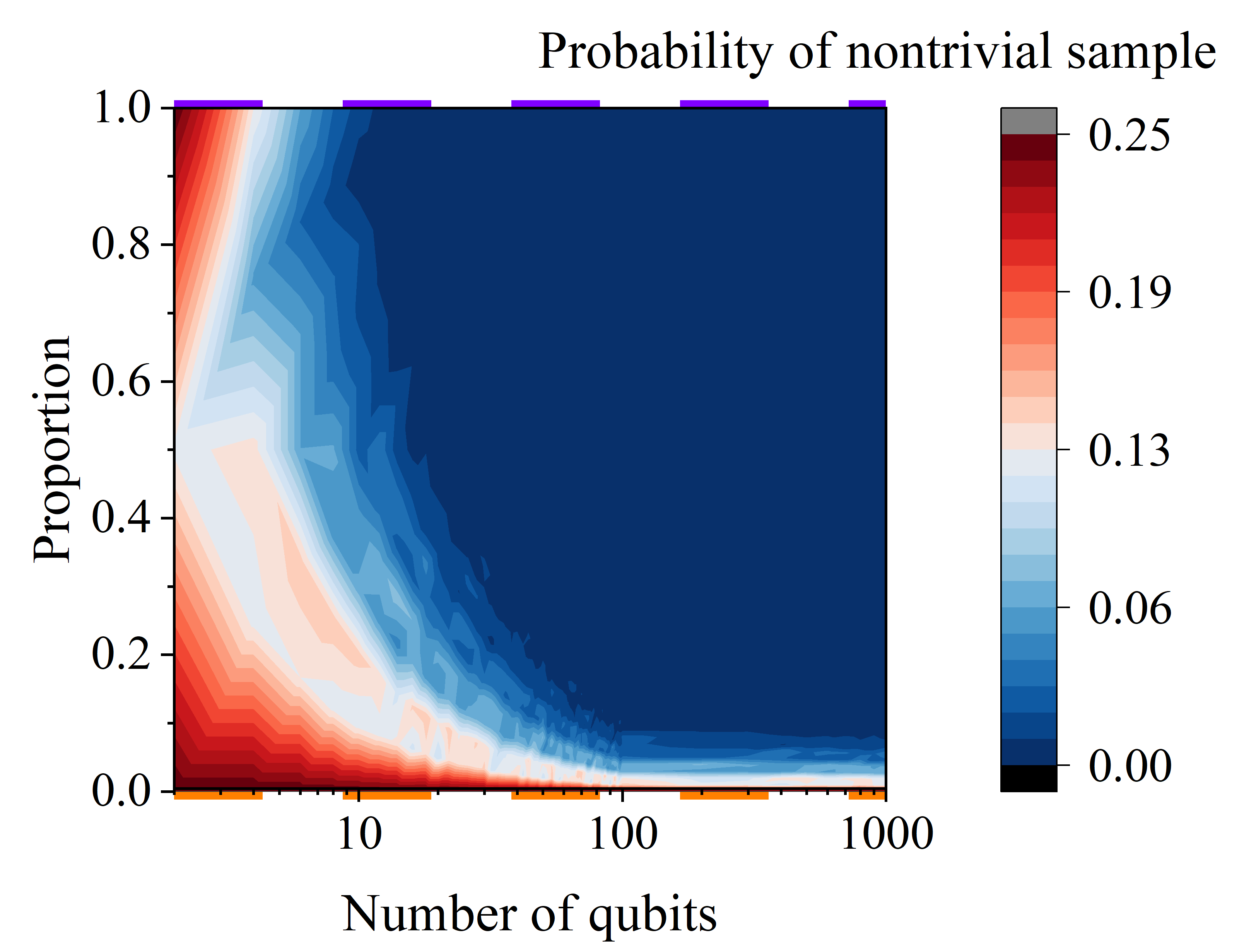}
        \caption{The number of qubits is even.}
        \label{fig:even_prob}
    \end{subfigure}
    \hfill
    \begin{subfigure}[b]{0.49\linewidth}
        \centering
        \includegraphics[width=\linewidth]{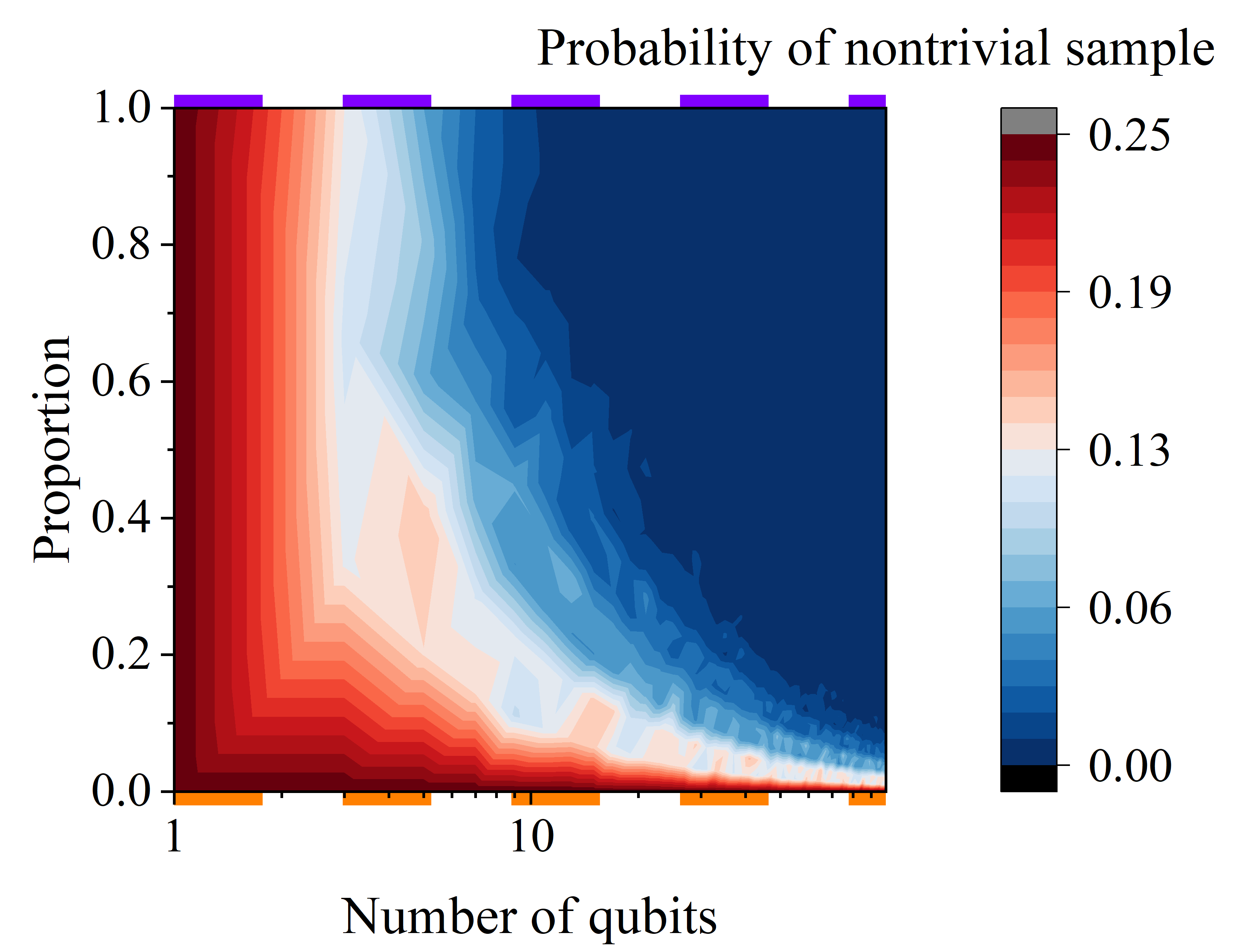}
        \caption{The number of qubits is odd.}
        \label{fig:odd_prob}
    \end{subfigure}
    \caption{Probability of nontrivial sampling for even and odd number of qubits. The $Z$-axis is the probability of nontrivial sampling, and the $Y$-axis refers to the proportion of $X$ operators in the observable. Therefore, at the points where \( x = 0 \), drawn by the orange line, the conditions degenerate into the strong stabilizer enhancement as stated in condition (\ref{condition2}). Similarly, at the points where \( x = 1 \), drawn by the purple line, the conditions degenerate into the weak stabilizer enhancement as described in condition (\ref{condition4}).}
    \label{fig:prob_sampling}
\end{figure}

Moreover, we present the results of the \texttt{stabilizer bootstrap} with the dataset and discuss the relationship between optimization efficiency and the size of the database. 

Our experiments include dataset sizes ranging from 100 to 1000 samples, and we set the number of features in each dataset equal to the number of qubits (1000). We utilize 3-layer stabilizer ans\"{a}tzes to improve the expressive capabilities of quantum circuits. Some other details including general numbers of layers are given in \ref{supp:C}. To assess the efficiency of our algorithm, we evaluate two key metrics: the minimum loss value achieved and the variance of loss values across all sampled points. 
A lower minimum loss value directly indicates better optimization performance. Additionally, higher variance during the sampling phase is beneficial for our search process, as it indicates the algorithm is exploring a diverse range of solutions. 
This exploration increases the probability of finding points with extreme loss values (both low and high), which is particularly valuable when seeking global minima, assuming the mean loss value is centered at zero. 
Our primary objective is to achieve the lowest possible loss value through this sampling and optimization process.

\begin{figure}[H]
    \centering
    \begin{subfigure}[b]{0.49\linewidth}
        \centering
        \includegraphics[width=\linewidth]{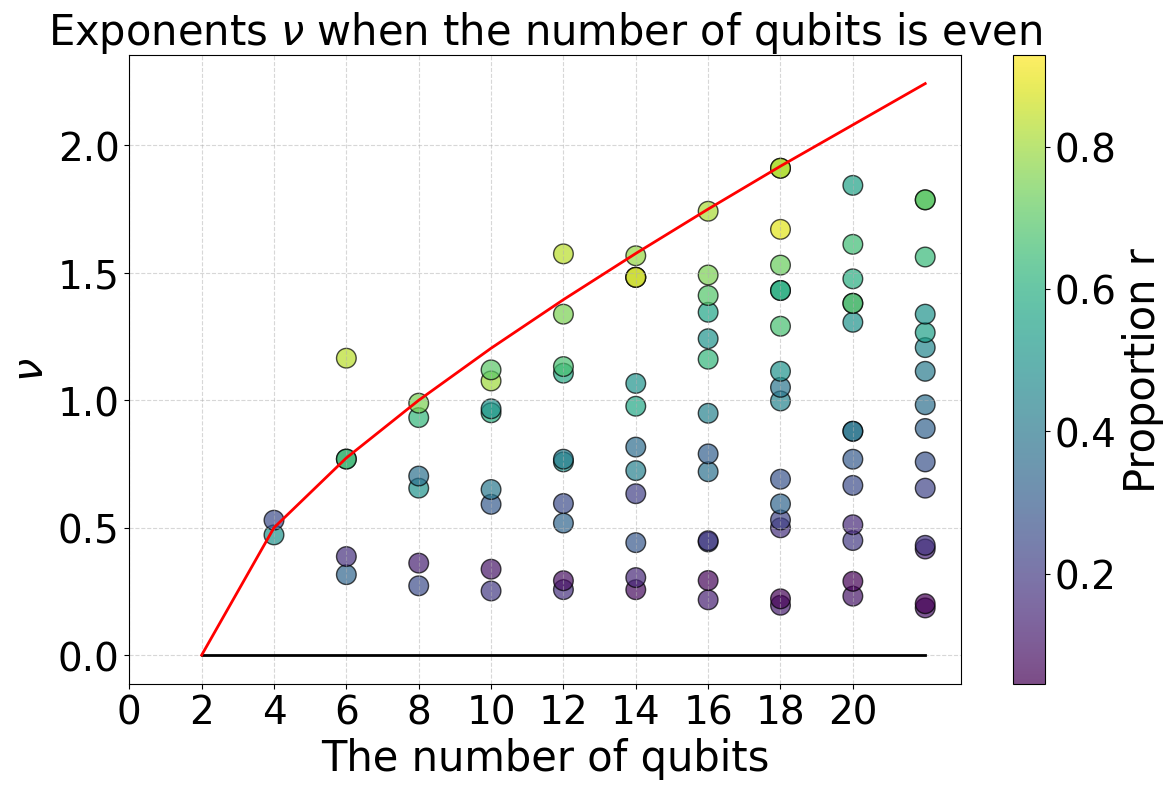}
        \caption{The number of qubits is even.}
     
    \end{subfigure}
    \hfill
    \begin{subfigure}[b]{0.49\linewidth}
        \centering
        \includegraphics[width=\linewidth]{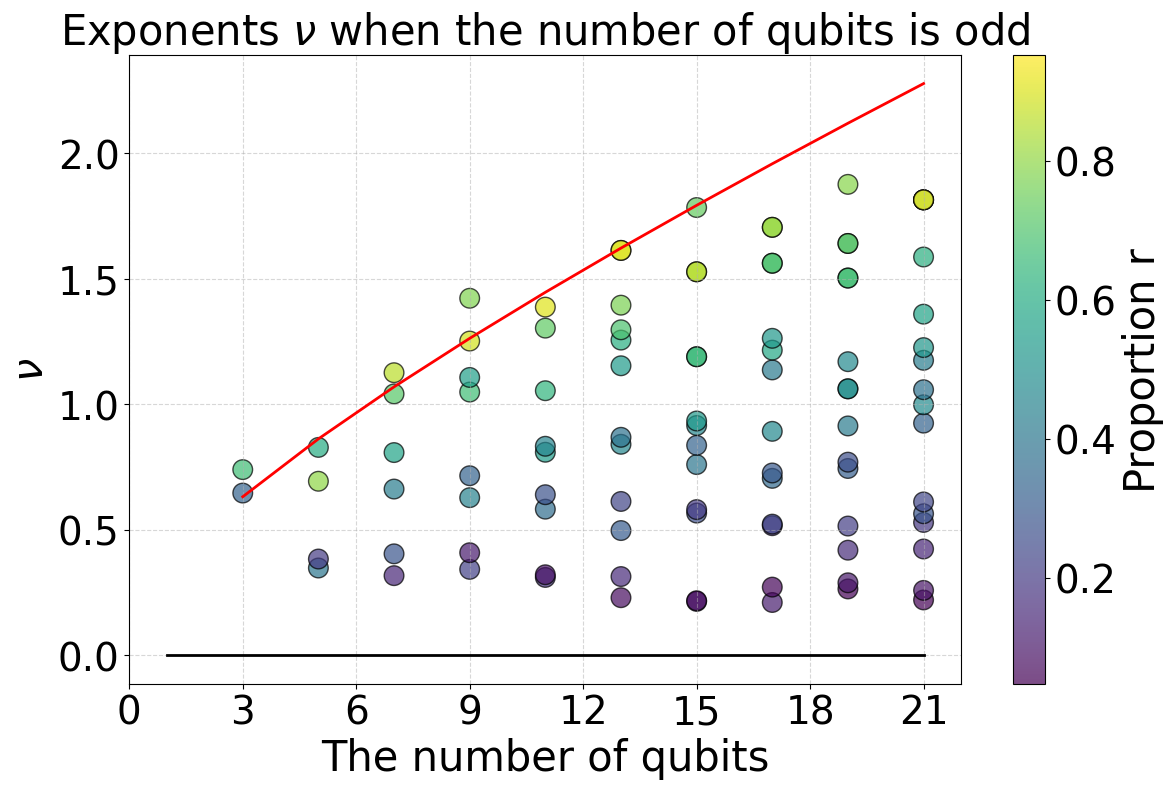}
        \caption{The number of qubits is odd.}
      
    \end{subfigure}
    \caption{Exponents for even and odd number of qubits. The Y-axis is the exponents, and the $X$-axis refers to the number of qubits. The red lines in these figures represent the exponent carve in Equation \ref{equ:logn} with $X$-string observables, while the black lines denote a constant value of \(\frac{1}{4}\) for $Z$-string observables in Equation \ref{equ:0}.}
    \label{fig:exponents}
\end{figure}

\begin{figure}
     \centering
     \includegraphics[width=\linewidth]{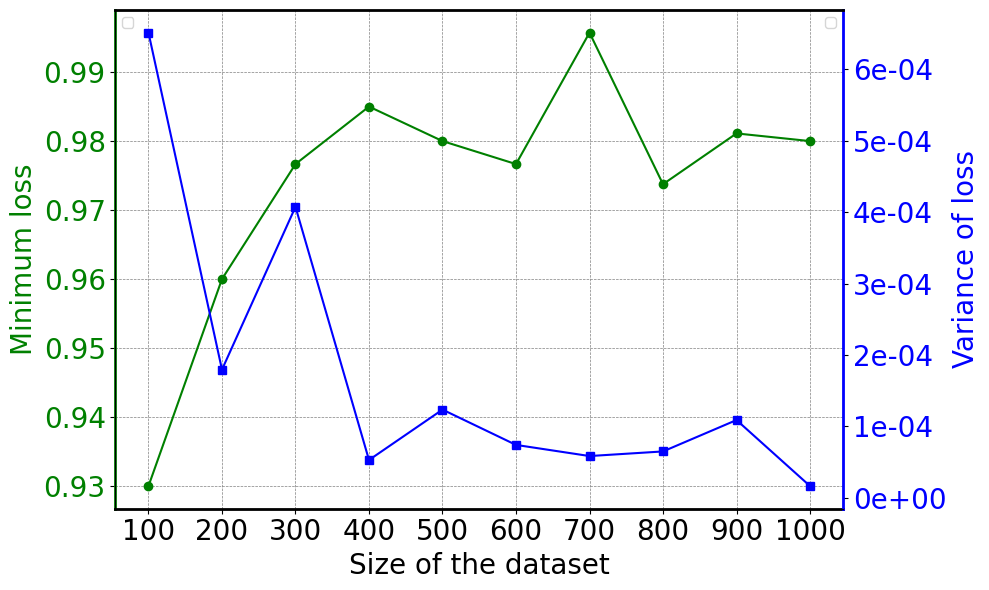}
     \caption{Search efficiency under different dataset sizes up to 1000 data.}
     \label{fig:experiment_dataset}
 \end{figure}
From FIG. \ref{fig:experiment_dataset}, we can draw the conclusion that as the dataset size increases, the search efficiency declines exponentially, characterized by higher loss values and reduced variance.

\bigskip 

\section{Discussion}
\label{discussion}

In this paper, we leverage the \texttt{stabilizer bootstrap} to enhance quantum machine learning (QML) circuits and analyze the challenges in the sampling stage of existing algorithms. Through extensive experiments, we examine the relationship between sampling efficiency and observables. Based on varying sampling efficiencies, we introduce the concepts of weak and strong stabilizer enhancement, demonstrating that our algorithm can scale to large qubit systems and larger datasets with the support of theoretical proofs (proof by induction) and high-performance computing. Additionally, we provide detailed mathematical proofs to substantiate key arguments. Our work not only establishes a practical framework for QML enhancement with numerical tests on systems of up to 10000 qubits but also proposes a novel methodology to quantify the spectrum between classical simulability and potential quantum advantages.

Our research also raises several new questions for QML enhancement. For instance, how can the general case for $0 < r < 1$ be rigorously proven, and how can this phenomenon be explained purely through the structure of the Clifford space? Furthermore, our current analysis of entanglement structures is limited to linear and reverse-linear configurations, which are \emph{dual} to each other by exchanging $X$ and $Z$ operators and interchanging strong and weak enhancements. This duality may be further explored using group theory. Additionally, future research could extend this analysis to explore the dependence on data encoding methods, structured datasets, the structures of more diverse variational circuits, the choice of measured operators, and the impact of QML loss functions. Finally, it will be interesting to study how our method could be used in the early fault-tolerant quantum computers, eventually towards Fault-tolerant Application Scale Quantum (FASQ) computing \cite{FASQ}, by looking at how our methods could be connected to stabilizer codes in quantum error correction. For instance, one can imagine using similar Bayesian optimization methods to optimize syndrome extraction circuits in this so-called partial syndrome measurement scheme \cite{berthusen2024partial}. These directions hold the potential to identify theoretical quantum advantages and address significant practical problems in QML or quantum computing in general.

\newpage
\textit{Acknowledgements.} YL, JC, XT, YZ and JL are supported in part by the University of Pittsburgh, School of Computing and Information, Department of Computer Science. This work is supported in part by NSF grants  \#2154973, \#1910413, \#2334628, and \#2312157. FTC is funded in part by the STAQ project under award NSF Phy-232580; in part by the US Department of Energy Office of Advanced Scientific Computing Research, Accelerated Research for Quantum Computing Program; and in part by the NSF Quantum Leap Challenge Institute for Hybrid Quantum Architectures and Networks (NSF Award 2016136), in part based upon work supported by the U.S. Department of Energy, Office of Science, National Quantum Information Science Research Centers, and in part by the Army Research Office under Grant Number W911NF-23-1-0077. The views and conclusions contained in this document are those of the authors and should not be interpreted as representing the official policies, either expressed or implied, of the U.S. Government. The U.S. Government is authorized to reproduce and distribute reprints for Government purposes notwithstanding any copyright notation herein. FTC is the Chief Scientist for Quantum Software at Infleqtion and an advisor to Quantum Circuits, Inc..

\bibliographystyle{unsrt}
\bibliography{Bibliography.bib}

\clearpage

\pagebreak

\onecolumngrid
\appendix

\vspace{0.5in}

\begin{center}
	{\Large \bf Appendix}
\end{center}

\section{Background}\label{supp:A}

\subsection{Quantum Machine Learning (QML)}
QML is a hybrid quantum-classical framework where variational quantum circuits are trained using classical optimization methods to minimize an objective function, as illustrated in FIG. \ref{fig:QML}. During each training iteration, the feature vectors $(x_1, \ldots, x_n)$ from the classical dataset are first encoded into a quantum state. This state is then conducted by a variational quantum circuit, and the final output state is subsequently measured. With all outcomes of the dataset, we calculate the loss function to evaluate the training efficiency. A classical optimizer is then employed to optimize the variational parameters, like gradient descent. When the loss function converges or achieves small enough, the training ends.
\begin{figure}[H]
    \centering
    \includegraphics[width=0.7\linewidth]{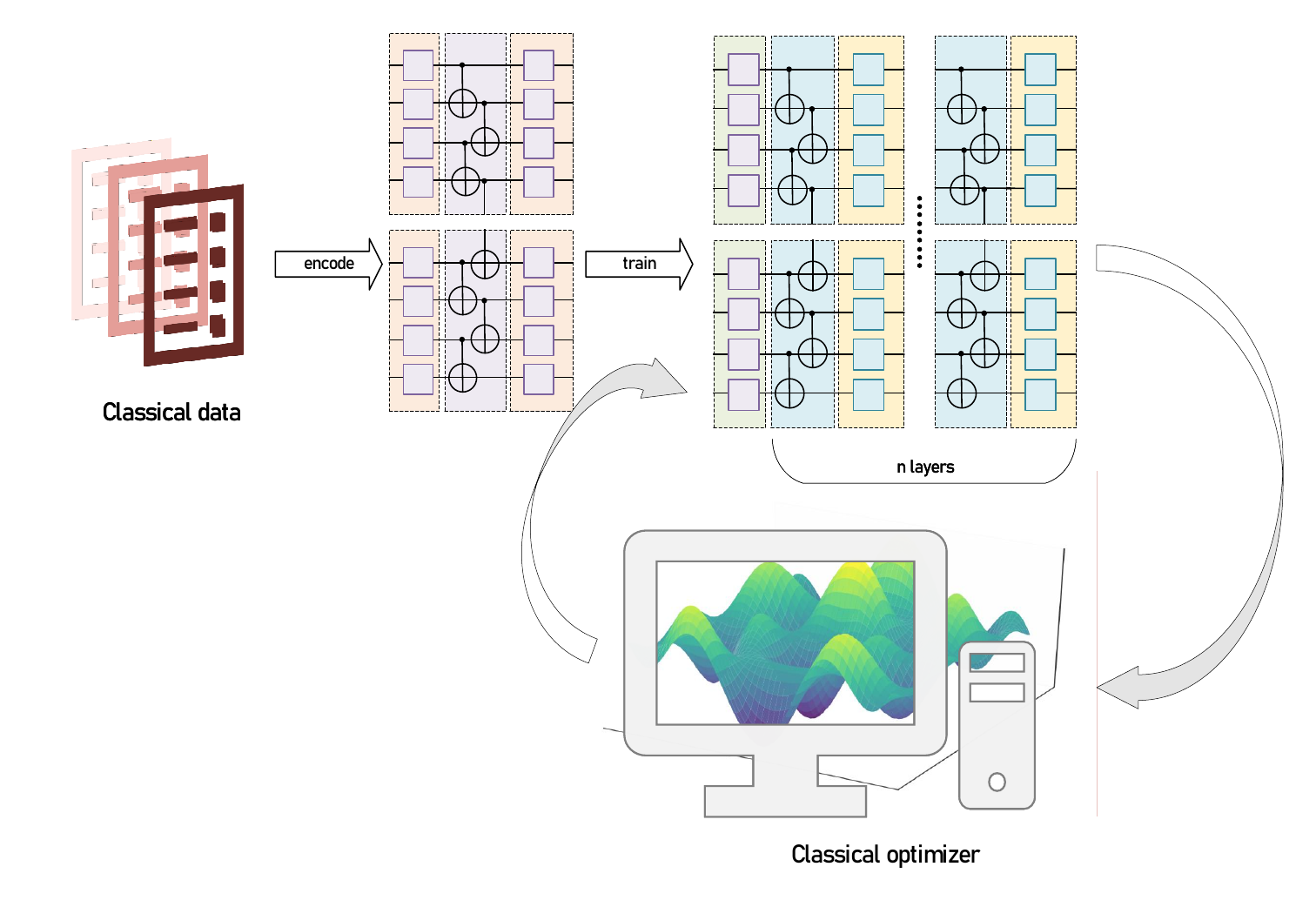}
    \caption{Overview of QML.}
    \label{fig:QML}
\end{figure}
In this study, we employ a \textit{make classification} function to generate a dataset \((x_{i1}, \ldots, x_{in}, y_i)\), where each \(x_{i1}\) represents a feature and \(y_i\) denotes the label. To ensure that the encoding quantum circuit is also a Clifford circuit, the features are transferred into  \(\{0, 1, 2, 3\}\) and subsequently mapped to the angles \(\{0, \pi, \pi/2, -\pi/2\}\), respectively. The loss function adopted is the mean squared error (MSE), which is defined as:
\[
\text{MSE} = \frac{1}{N} \sum_{i=1}^{N} \left( y_i - \hat{y}_i \right)^2,
\]
where $\hat{y}_i$ denotes the predicted label, $y_i$ represents the true label, and N is the size of the dataset.
\subsection{Bayesian optimization}
Bayesian Optimization (BO) is an efficient framework for the global optimization of expensive black-box functions, particularly in scenarios where function evaluations are costly, time-consuming, or require significant computational resources. Unlike traditional optimization methods that rely on gradient, BO is well-suited for optimizing objective functions that are denoted in discrete space, making it suitable for this task in Clifford space.

\begin{figure}
    \centering
    \includegraphics[width=0.7\linewidth]{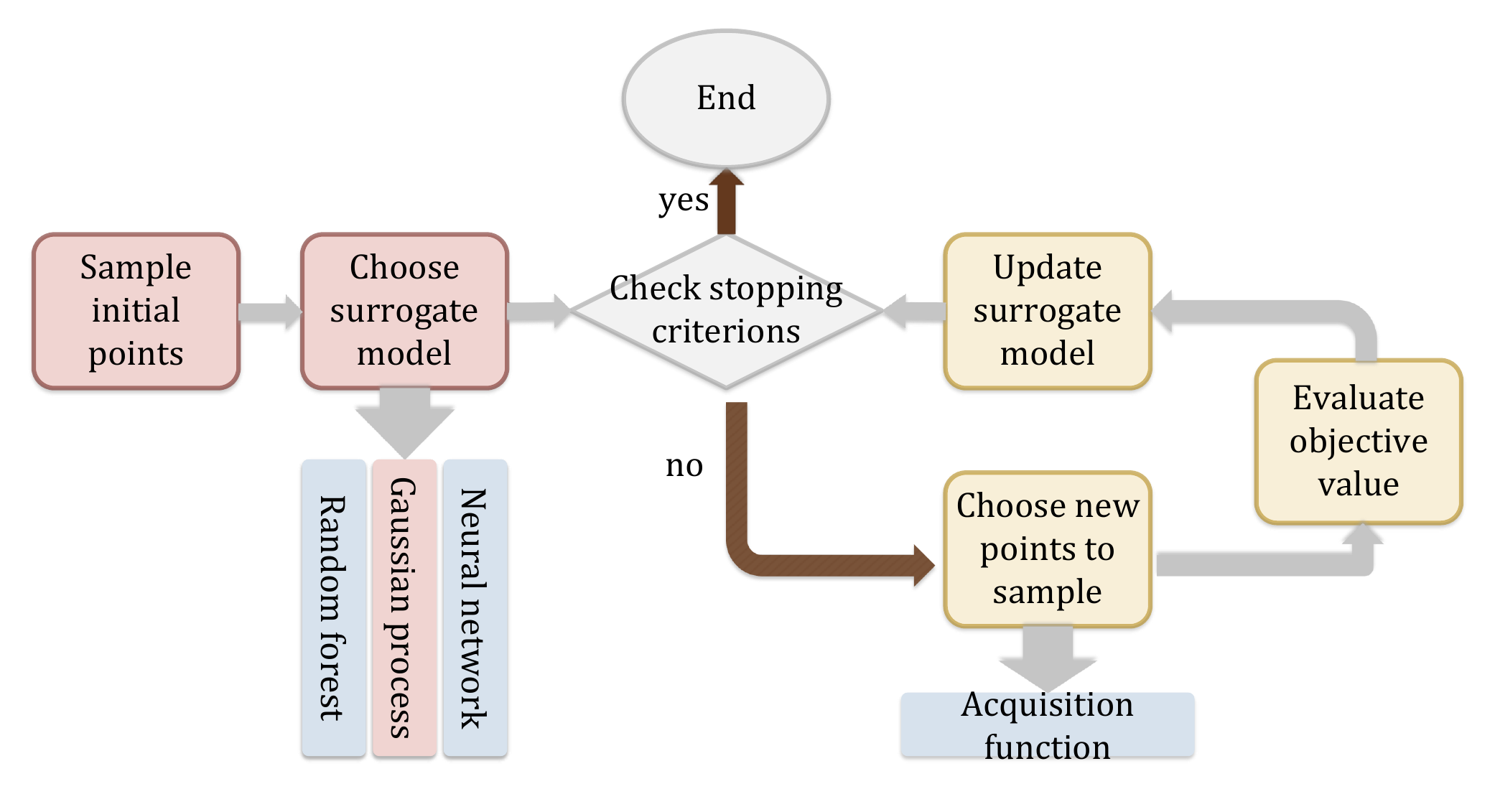}
    \caption{Overview of Bayesian optimization.}
    \label{fig:BO}
\end{figure}

We draw a flow chart for BO in FIG. \ref{fig:BO}. Now we will explain why the sampling is so vital. Assume we use the Gaussian process (GP) as the surrogate model and Expected improvement (EI) as the acquisition function. 
The following discussion is based entirely on these settings. 
The optimization chooses the next points according to EI. Based on mean value $\mu{(x)}$ and standard deviation $\sigma{(x)}$, we can calculate EI by:
$$
\operatorname{EI}(x) = (\mu(x) - f^*) \Phi(Z) + \sigma(x) \phi(Z), \label{eq:EI}
$$
where \( Z = \frac{\mu(x) - f^*}{\sigma(x)} \), \( \Phi(Z) \) is the cumulative distribution function (CDF) of the standard normal distribution, \( \phi(Z) \) is the probability density function (PDF) of the standard normal distribution and $f^*$ is the current optimal value. What would $\mu{(x)}$ and $\sigma{(x)}$ be when sampling points are all trivial? The expressions of $\mu{(x)}$ and $\sigma{(x)}$ are:
\[
\mu(x) = k(X, x)^T [K(X, X) + \sigma_n^2 I]^{-1} y,
\]
\[
\sigma^2(x) = k(x, x) - k(X, x)^T [K(X, X) + \sigma_n^2 I]^{-1} k(X, x),
\]
where $k(x,x)$, $k(X,x)$ and $k(X,X)$ are all kernel function, $y$ is the observed value of sampled data. Setting the noise parameter $\sigma_n^2$ as 0, we could obtain that $\mu(x)$ is 0. The value \( \sigma(x) \) represents the \textbf{uncertainty} in the prediction at the point \( x \). When \( x \) is close to the training data points, the covariance \( k(X, x) \) is relatively large, which results in a smaller variance \( \sigma^2(x) \). Conversely, when \( x \) is far away from the training data points, the variance \( \sigma^2(x) \) increases, indicating greater uncertainty in the prediction. Bringing these results back to the original formula \ref{eq:EI}, we could know $Z$ is 0 and EI becomes:
$$
\operatorname{EI}(x) = \sigma(x) \phi(0)=\frac{1}{\sqrt{2 \pi}}  \sigma(x).
$$
Therefore, according to EI, the optimization could only choose points far away from the existing points until it finds an untrivial point. Considering the efficiency of these two steps, i.e., the sampling is easy to complete parallelly with multi-CPUs, we choose to focus on improving sampling efficiency in this work.

\subsection{The stabilizer formalism}

Classical quantum computing tasks often require exponential resources. However, according to the Gottesman-Knill theorem, Clifford circuits can be efficiently simulated by classical computers in polynomial time. The Gottesman-Knill theorem states that any quantum circuit composed exclusively of Clifford gates with the preparation of qubits in computational basis states and measurements in the computational basis, can be efficiently simulated on a classical computer. To define Clifford operators, we first introduce the Pauli group $P_1$ on a single qubit. The group $P_1$ is generated by the operators $\{i, \sigma_x, \sigma_z\}$ under multiplication, and it contains 16 elements:
\[
\{\pm I, \pm \sigma_x, \pm \sigma_y, \pm \sigma_z, \pm iI, \pm i\sigma_x, \pm i\sigma_y, \pm i\sigma_z\}.
\]

The Pauli group $P_n$ on $n$ qubits is constructed by applying $P_1$ to each qubit. That is, $P_n$ is generated by:
\[
\langle i, \sigma_x^{(1)}, \sigma_x^{(2)}, \ldots, \sigma_x^{(n)}, \sigma_y^{(1)}, \ldots, \sigma_y^{(n)}, \sigma_z^{(1)}, \ldots, \sigma_z^{(n)} \rangle.
\]

The Clifford group corresponding to the $n$-qubit Pauli group $P_n$ is defined as follows:

\textbf{Definition (Clifford group):}  
If for an $n$-qubit operator $g \in SU(2^n)$, it holds that for every element $p \in P_n$,
\[
g p g^{-1} \in P_n,
\]
then $g$ is called an $n$-qubit Clifford operator. All such Clifford operators $g$ form a group structure known as the $n$-qubit Clifford group. Therefore, Clifford gates include the single-qubit gates
$
\sigma_x,  \sigma_z,  \sqrt{i\,\sigma_x},  \sqrt{i\,\sigma_y},  \sqrt{i\,\sigma_z},  H,  S $
and the two-qubit CNOT gate. Among these gates, $S$, $H$, and CNOT can be combined to produce all Clifford gates. It is noteworthy that the $T$ gate does not belong to the Clifford group, yet $H$, $T$, and CNOT together can achieve universal quantum computation.

\section{Proof for Main Theorems}\label{supp:B}

The formal theorems supporting Table \ref{tab:prob} are given in this section.

\begin{theorem}
	For an n-qubit circuit with the reverse linear entanglement structure in FIG.\ref{fig:clifford_ansatz} (composed by one layer of $R_y$ gates and	one layer of reverse linear CNOT gates), where the angles of the $R_y$ gates are restricted to the set $\{0, \pi, \frac{\pi}{2}, - \frac{\pi}{2} \}$, and the observable is a Pauli-$Z$ string, the probability that the measurement	outcome is $1$ or $- 1$ is both $\frac{1}{4}$, while the probability that the outcome is $0$ is $\frac{1}{2}$.                    \label{theorem:1}                       
\end{theorem}

\begin{proof}
\begin{figure}[H]
    \centering
    \includegraphics[width=0.75\linewidth]{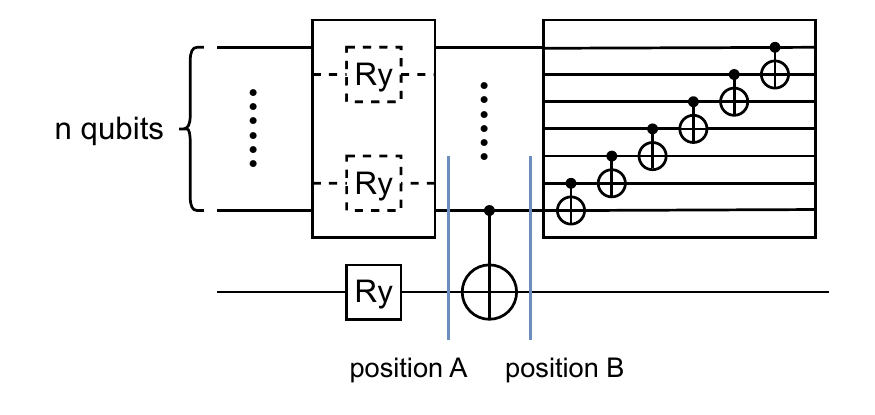}
    \caption{Proof details in Theorem \ref{theorem:1} with Position A and Position B.}
    \label{fig:proof1}
\end{figure}
    We use mathematical induction to prove the theorem:

For the sake of discussion, we define the states in the eigenspace of \( Z \otimes \cdots \otimes Z \) with eigenvalue 1 as \emph{$1$-states}, those in the eigenspace with eigenvalue $-1$ as \emph{$-1$-states}, and the remaining states as \emph{$0$-states}. Since the circuit is a Clifford circuit, the measurement outcome for \emph{$1$-states} is 1, for \emph{$-1$-states} is $-1$, and for all other states is 0.

First, for the case where the number of qubits is 1, i.e., $n=1$, the circuit is equivalent to a single $R_y$ gate. The $R_y$ gate has an equal probability of taking values from the set $\{0, \pi, \frac{\pi}{2}, -\frac{\pi}{2}\}$, and the corresponding states are $|0\rangle$, $|1\rangle$, $|+\rangle$, and $|-\rangle$, respectively. It is easy to see that the above conclusion holds in this case.

Now, suppose the conclusion holds for $n = k$. Then, for $n = k+1$:

(1) If the $n$-th qubit is $|0\rangle$ or $|1\rangle$, then the first $n$ qubits and the $(n+1)$-th qubit are not entangled, so their measurement outcomes are independent. The probability that the final $n+1$ qubits are in the $1$-state corresponds to the case where both the first $n$ qubits and the $(n+1)$-th qubit are either in the $1$-state or both in the $-1$-state.

Here, it is important to note that 
\[
p \text{ }(\text{the first } n \text{ qubits are in the } 1\text{-state}) = p\text{ }(\text{the first } n \text{ qubits are in the } 1\text{-state} | \text{the } n\text{-th qubit is } |0\rangle \text{ or } |1\rangle)
\]

It suffices to state that the $n$-th qubit being $|0\rangle$ or $|1\rangle$ is a necessary condition for the first $n$ qubits to be in the $1$-state.

If the $n$-th qubit is $|0\rangle \pm |1\rangle$:
\[
|\phi_B\rangle = \sum_{i=1}^k |x_{i_1} \dots x_{i_{n-1}}0\rangle \pm \sum_{i=1}^k |x_{i_1} \dots x_{i_{n-1}}1\rangle \rightarrow |\phi_{\text{final}}\rangle = \sum_{i=1}^k |x_{i_1}' \dots x_{i_{n-1}}'x_n'\rangle + \sum_{i=1}^k |x_{i_1}' \dots x_{i_{n-1}}{\overline{x_n'}}\rangle.
\]

If $\sum_{i=1}^k |x_{i_1}' \dots x_{i_{n-1}}'x_n'\rangle$ is in the $1$-state, then $\sum_{i=1}^k |x_{i_1}' \dots x_{i_{n-1}}{\overline{x_n'}}\rangle$ must be in the $-1$-state, and vice versa. If $\sum_{i=1}^k |x_{i_1}' \dots x_{i_{n-1}}'x_n'\rangle$ is in the $0$-state, suppose there are $m$ qubits in the $1$-state and $k-m$ qubits in the $-1$-state. Then, $\sum_{i=1}^k |x_{i_1}' \dots x_{i_{n-1}}{\overline{x_n'}}\rangle$ is also in the $0$-state, with $k-m$ qubits in the $1$-state and $m$ qubits in the $-1$-state. In this case, $|\phi_{\text{final}}\rangle$ has $k$ qubits in the $1$-state and $k$ qubits in the $-1$-state, meaning that $|\phi_{\text{final}}\rangle$ must be in the $0$-state. In conclusion, 
\[
p\text{ }(\text{the first } n \text{ qubits are in the } 1\text{-state}) = p\text{ }(\text{the first } n \text{ qubits are in the } 1\text{-state} | \text{the } n\text{-th qubit is } |0\rangle \text{ or } |1\rangle) = \frac{1}{4}.
\]

Thus, if the $n$-th qubit is in the state $|0\rangle$ or $|1\rangle$, the probability that the $n+1$ qubits are in the $1$-state or $-1$-state is:

\[
\frac{1}{4} \times \frac{1}{4} + \frac{1}{4} \times \frac{1}{4} = \frac{1}{8}.
\]
(2) If the $n$-th qubit is in the state $|0\rangle \pm |1\rangle$, and the $(n+1)$-th qubit is in the state $|0\rangle$ or $|1\rangle$, then the first $n$ qubits and the $(n+1)$-th qubit become entangled. Let the quantum state at Position $A$ be:

\[
|\phi_A\rangle = \sum_{i=1}^k |x_{i_1} \dots x_{i_{n-1}}\rangle \otimes (|0\rangle + |1\rangle) \otimes |x_n\rangle
\]

where $x_{i_j}$ takes values 0 or 1, and $|x_{i_1} \dots x_{i_{n-1}}\rangle$ is either in the $1$-state or the $-1$-state.

The quantum state at Position $B$ is:

\[
|\phi_B\rangle = \sum_{i=1}^k |x_{i_1} \dots x_{i_{n-1}}\rangle \otimes (|0x_n\rangle + |1(x_n)^{\overline{}}\rangle) = \sum_{i=1}^k |x_{i_1} \dots x_{i_{n-1}}0x_{n+1}\rangle + \sum_{i=1}^k |x_{i_1} \dots x_{i_{n-1}}1{\overline{x_{n+1}}}\rangle.
\]

The final quantum state is:

\[
|\phi_{\text{final}}\rangle = \sum_{i=1}^k |x_{i_1}' \dots x_{i_{n-1}}'x_n'x_{n+1}\rangle + \sum_{i=1}^k |x_{i_1}' \dots x_{i_{n-1}}{\overline{x_n'}}{\overline{x_{n+1}}}\rangle.
\]

(a) If both parts are in the $1$-state, then they are both in the $1$-state. If both are in the $-1$-state, they are both in the $-1$-state. In this case, we only need to discuss the probability that $\sum_{i=1}^k |x_{i_1}' \dots x_{i_{n-1}}'x_n'x_{n+1}\rangle$ is in the $1$-state, which is easily obtained as $1/8$.

(b) If $\sum_{i=1}^k |x_{i_1}' \dots x_{i_{n-1}}'x_n'x_{n+1}\rangle$ is in the $0$-state, then $\sum_{i=1}^k |x_{i_1} \dots x_{i_{n-1}}x_nx_{n+1}\rangle$ must also be in the $0$-state. (If $\sum_{i=1}^k |x_{i_1} \dots x_{i_{n-1}}x_nx_{n+1}\rangle$ is in the $1$-state or $-1$-state, since $\sum_{i=1}^k |x_{i_1} \dots x_{i_{n-1}}x_nx_{n+1}\rangle$ is not entangled, it must be that $k=1$.) In this case, we can assume:
\[
\sum_{i=1}^k |x_{i_1} \dots x_{i_{n-1}}x_nx_{n+1}\rangle = |x_1\rangle \otimes |x_2\rangle \otimes \dots \otimes |+\rangle \otimes \dots \otimes |-\rangle \otimes \dots \otimes |x_n\rangle, \quad x_i \in \{0,1\}.
\]

Thus, it is easy to know $k$ is an even number, and there are $k/2$ qubits in the $1$-state and $k/2$ qubits in the $-1$-state. Now, we need to show that in $\sum_{i=1}^k |x_{i_1}' \dots x_{i_{n-1}}'x_n'x_{n+1}\rangle$, the number of $1$-states and $-1$-states is also equal.

Let’s assume that at Position $A$ in FIG. \ref{fig:proof1}, we have $d$ qubits $\{k_1, k_2, \dots, k_d\} $ which are in a superposition state $|+\rangle$ or $|-\rangle$. Then,

$$
\sum_{i=1}^{k} |x_{i_1} \dots x_{i_{n-1}} x_n x_{n+1} \rangle = \sum_{k_j \in \{0, 1\}} |x_1 x_2 \dots x_{k_1} \dots x_{k_2} \dots x_{k_d} \dots x_{n+1} \rangle
$$

It is easy to see that we can pair the \(2^d\) basis states, where each pair differs by only one element. In each pair of basis states, one is $1$-state and the other is $-1$-state. Therefore, after passing through the CNOT layer, in each pair of basis states, there will still be one $1$-state one $-1$-state. Thus, in the expression

$$
\sum_{i=1}^{k} |x_{i_1}' \dots x_{i_{n-1}}' x_n' x_{n+1}' \rangle
$$

the number of $1$-states and  $-1$-states remains equal, meaning that the final state \( |\varphi_{\text{final}}\rangle \) will necessarily be the  $0$-state.

(3) If the \(n\)-th qubit is in the state \( |0\rangle \pm |1\rangle \), and the \( (n+1)\)-th qubit is also in the state \( |0\rangle \pm |1\rangle \), then there is no entanglement between the \(n\)-th and \( (n+1)\)-th qubits, and the state must also be $0$-state.

Thus, in conclusion, the probability that the \(n+1\) qubits are in the $1$-state or  $-1$-state is also \( \frac{1}{4} \).

\end{proof}

\begin{theorem}
	For an n-qubit circuit with the linear entanglement structure in FIG.\ref{fig:clifford_ansatz} \ (composed by one layer of $R_y$ gates	and one layer of \ linear CNOT gates), where the angles of the $R_y$ gates are	restricted to the set $\{0, \pi, \frac{\pi}{2}, - \frac{\pi}{2} \}$, and the	observable is a Pauli-$Z$ string, the probability that the measurement outcome	is $1$ or $- 1$ is both $\frac{1}{2^{\lceil n / 2 + 1 \rceil}}$, while the probability that the outcome is $0$ is $1 - \frac{1}{2^{\lceil n / 2\rceil}}$. 
    \label{theorem:2}
\end{theorem}

\begin{proof}
\begin{figure}[H]
    \centering
    \includegraphics[width=0.75\linewidth]{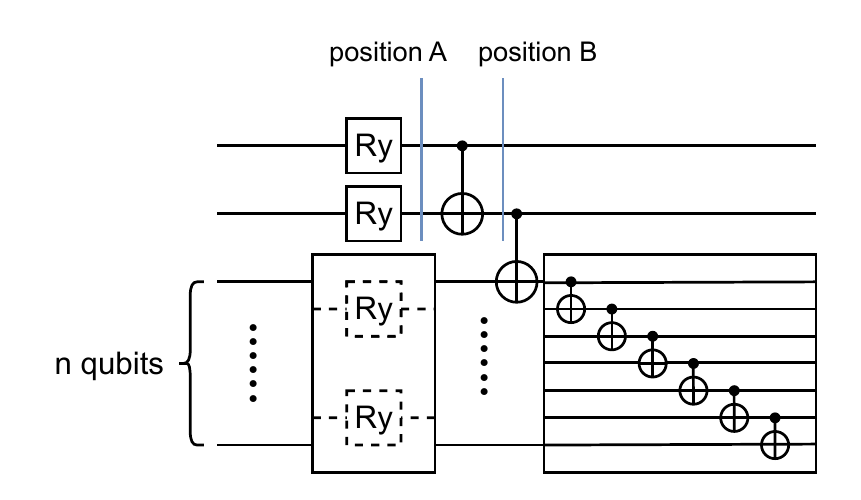}
    \caption{Proof details in Theorem \ref{theorem:2} with Position $A$ and Position $B$. }
    \label{fig:proof2}
\end{figure}
    We will use mathematical induction to complete the proof. Let \( n = 2k + 1 \), in which case the probability of the $1$-state and $-1$-state is \( \frac{1}{2^{k+2}} \). We need to prove that when \( n = 2k + 3 \), the probabilities of the $1$-state and $-1$-state are both \( \frac{1}{2^{k+3}} \).

First, the case \( n = 1 \) is easily verified to be true, so it is omitted.

Assume that for \( n = 2k + 1 \), the probabilities of the $1$-state and $-1$-state are both \( \frac{1}{2^{k+2}} \). Now consider the case when \( n = 2k + 3 \):

(1) If the first two qubits take \(\{|00\rangle, |11\rangle, |10\rangle, |01\rangle\}\), then the first two qubits and the remaining \(n\) qubits are not entangled. In this case, the probability of the system being in the $1$-state is determined by the probability that the first two qubits are in the $1$-state, and the remaining \(n\) qubits are either all in the $1$-state or all in the $-1$-state. For the probability that the first two qubits are in the $1$-state, we can refer to the additional TABLE \ref{tab:added_table}, which is easily found to be \( \frac{1}{4} \). Therefore, the probability that these \(n+1\) qubits are in the $1$-state is:

$$
\frac{1}{8} \cdot \frac{1}{2^{k+2}} + \frac{1}{8} \cdot \frac{1}{2^{k+2}} = \frac{1}{2^{k+4}}
$$

(2) If the first two qubits take \(\{|00\rangle \pm |11\rangle, |01\rangle \pm |10\rangle\}\), first consider the case of \(|00\rangle \pm |11\rangle\), then at position \(B\) in FIG. \ref{fig:proof2}:

$$
|\varphi_B\rangle = \sum_{i=1}^{k} |00 x_{i_3} x_{i_4} \dots x_{i_{n+2}} \rangle + \sum_{i=1}^{k} |11 x_{i_3} x_{i_4} \dots x_{i_{n+2}} \rangle
$$

Then, finally:

$$
|\varphi_{\text{final}}\rangle = \sum_{i=1}^{k} |00 x_{i_3}' x_{i_4}' \dots x_{i_{n+2}}' \rangle + \sum_{i=1}^{k} |11 x_{i_3}'\overline{x_{i_4}'} \dots \overline{x_{i_{n+2}}'} \rangle
$$

Since \(n\) is odd, \(|\varphi_{\text{final}}\rangle\) must be the  $0$-state. Similarly, when considering \(|01\rangle \pm |10\rangle\), \(|\varphi_{\text{final}}\rangle\) is also the $0$-state.

(3) If the first two qubits take the state \( |\psi\rangle \otimes (|0\rangle \pm |1\rangle) \), then at position \( B \):

$$
|\varphi_B\rangle = |\psi\rangle \otimes \left( \sum_{i=1}^{k} |0 x_{i_3} x_{i_4} \dots x_{i_{n+2}} \rangle + \sum_{i=1}^{k} |1 x_{i_3} x_{i_4} \dots x_{i_{n+2}} \rangle \right)
$$

Finally:

$$
|\varphi_{\text{final}}\rangle = |\psi\rangle \otimes \left( \sum_{i=1}^{k} |0 x_{i_3}' x_{i_4}' \dots x_{i_{n+2}}' \rangle + \sum_{i=1}^{k} |1 x_{i_3}' \overline{x_{i_4}'} \dots \overline{x_{i_{n+2}}'} \rangle \right)
$$

Since \( |0 x_{i_3}' x_{i_4}' \dots x_{i_{n+2}}'\rangle \) has an even number of qubits, both \( |0 x_{i_3}' x_{i_4}' \dots x_{i_{n+2}}'\rangle \) and \( |1 x_{i_3}' \overline{x_{i_4}'} \dots \overline{x_{i_{n+2}}'} \rangle \) will either be in the $1$-state or the $-1$-state.

Therefore, the probability of the final state being in the $1$-state is:

$$
\frac{1}{4} \cdot \frac{1}{2} \cdot \frac{1}{2^{k+2}} + \frac{1}{4} \cdot \frac{1}{2} \cdot \frac{1}{2^{k+2}} = \frac{1}{2^{k+4}}
$$

In conclusion, when \( n = 2k + 3 \), the probability of the system being in the $1$-state is:

$$
\frac{1}{2^{k+4}} + \frac{1}{2^{k+4}} = \frac{1}{2^{k+3}}
$$

Thus, the case where \(n\) is odd is proven. Below is the case where \(n\) is even:

When \(n = 2\), the conclusion is easily derived, so it is omitted. When \(n = 2k\), the probabilities of the $1$-state and $-1$-state are both \(\frac{1}{2^{k+1}}\). Now, consider the case when \(n = 2k + 2\):

(1) If the first two qubits take the states \(\{|00\rangle, |11\rangle, |10\rangle, |01\rangle\}\), then the first two qubits and the remaining \(n\) qubits are not entangled. In this case, the probability of $n+2$ qubits being in the $1$-state is determined by the probability that the first two qubits are in the $1$-state and the remaining \(n\) qubits are either all in the $1$-state or all in the $-1$-state. For the probability that the first two qubits are in the $1$-state, we can refer to the additional table\ref{add_table}, which is easily found to be \(\frac{1}{4}\). Therefore, the probability that these \(n+1\) qubits are in the $1$-state is:

$$
\frac{1}{8} \cdot \frac{1}{2^{k+2}} + \frac{1}{8} \cdot \frac{1}{2^{k+2}} = \frac{1}{2^{k+4}}
$$

(2) If the first two qubits take \(\{|00\rangle \pm |11\rangle, |01\rangle \pm |10\rangle\}\), first consider the case of \(|00\rangle \pm |11\rangle\), then at position \(B\) in FIG. \ref{fig:proof2}:

$$
|\varphi_B\rangle = \sum_{i=1}^{k} |00 x_{i_3} x_{i_4} \dots x_{i_{n+2}} \rangle + \sum_{i=1}^{k} |11 x_{i_3} x_{i_4} \dots x_{i_{n+2}} \rangle
$$
Finally,

$$
|\varphi_{\text{final}}\rangle = \sum_{i=1}^{k} |00 x_{i_3}' x_{i_4}' \dots x_{i_{n+2}}' \rangle + \sum_{i=1}^{k} |11 x_{i_3}' \overline{x_{i_4}'} \dots \overline{x_{i_{n+2}}'} \rangle
$$

Since \(n\) is even, both \( |00 x_{i_3}' x_{i_4}' \dots x_{i_{n+2}}'\rangle \) and \( |11 x_{i_3}' \overline{x_{i_4}'} \dots \overline{x_{i_{n+2}}'}\rangle \) will either be in the $1$-state or in the $-1$-state. Therefore, the probability of the final state being in the $1$-state is:

$$
\frac{1}{4} \cdot \frac{1}{2} \cdot \frac{1}{2^{k+2}} + \frac{1}{4} \cdot \frac{1}{2} \cdot \frac{1}{2^{k+2}} = \frac{1}{2^{k+4}}
$$

(3) If the first two qubits take the state \( |\psi\rangle \otimes (|0\rangle \pm |1\rangle) \), then at position \(B\):

$$
|\varphi_B\rangle = |\psi\rangle \otimes \left( \sum_{i=1}^{k} |0 x_{i_3} x_{i_4} \dots x_{i_{n+2}} \rangle + \sum_{i=1}^{k} |1 x_{i_3} x_{i_4} \dots x_{i_{n+2}} \rangle \right)
$$

Finally,

$$
|\varphi_{\text{final}}\rangle = |\psi\rangle \otimes \left( \sum_{i=1}^{k} |0 x_{i_3}' x_{i_4}' \dots x_{i_{n+2}}' \rangle + \sum_{i=1}^{k} |1 x_{i_3}' \overline{x_{i_4}'} \dots \overline{x_{i_{n+2}}'} \rangle \right)
$$

Since \( |0 x_{i_3}' x_{i_4}' \dots x_{i_{n+2}}'\rangle \) has an odd number of qubits, the final state \( |\varphi_{\text{final}}\rangle \) must be the  $0$-state.

In conclusion, when \(n = 2k + 2\), the probability of the $n+2$ qubits being in the $1$-state is:

$$
\frac{1}{2^{k+4}} + \frac{1}{2^{k+4}} = \frac{1}{2^{k+3}}
$$

\end{proof}

\begin{table}[h!]
\centering
\begin{tabular}{c|c|c}
\hline
&\textbf{Position $A$} & \textbf{Position $B$} \\
\hline
\text{$1$-state} & 
$|00\rangle, |11\rangle$ & 
$|00\rangle, |11\rangle, \frac{1}{\sqrt{2}}|00\rangle \pm \frac{1}{\sqrt{2}}|11\rangle$ \\
\hline
\text{$-1$-state} & 
$|01\rangle, |10\rangle$ & 
$|01\rangle, |10\rangle, \frac{1}{\sqrt{2}}|01\rangle \pm \frac{1}{\sqrt{2}}|10\rangle$ \\
\hline
\text{$0$-state} & 
$\frac{1}{\sqrt{2}}(|00\rangle \pm |01\rangle), \frac{1}{\sqrt{2}}(|10\rangle \pm |11\rangle), \frac{1}{\sqrt{2}}(|00\rangle \pm |10\rangle), \frac{1}{\sqrt{2}}(|01\rangle \pm |11\rangle)$ & 
$\frac{1}{\sqrt{2}}(|00\rangle \pm |01\rangle), \frac{1}{\sqrt{2}}(|10\rangle \pm |11\rangle), \frac{1}{2}(|0\rangle \pm |1\rangle) \otimes (|0\rangle \pm |1\rangle)$ \\
\hline
\end{tabular}

\caption{State transitions from Position $A$ to Position $B$ for 2 qubits.}
\label{tab:added_table}
\end{table}

Now, by symmetry, we can naturally derive the following two further theorems:

\begin{theorem}
	For an $n$-qubit circuit with reverse linear entanglement structure in FIG.\ref{fig:clifford_ansatz} (composed by one layer of $R_y$ gates and	one layer of reverse linear CNOT gates), where the angles of the $R_y$ gates are restricted to the set $\{0, \pi, \frac{\pi}{2}, - \frac{\pi}{2} \}$, and the observable is a Pauli-$X$ string, the probability that the measurement	outcome is $1$ or $- 1$ is both $\frac{1}{2^{\lceil n / 2 + 1 \rceil}}$, while the probability that the outcome is $0$ is $1 - \frac{1}{2^{\lceil n / 2\rceil}}$.
\end{theorem}

\begin{theorem}
	For an $n$-qubit circuit with the linear entanglement structure in FIG.\ref{fig:clifford_ansatz} \ (composed by one layer of $R_y$ gates	and one layer of \ linear CNOT gates), where the angles of the $R_y$ gates are	restricted to the set $\{0, \pi, \frac{\pi}{2}, - \frac{\pi}{2} \}$, and the	observable is a Pauli-$X$ string, the probability that the measurement outcome	is $1$ or $- 1$ is both $\frac{1}{4}$, while the	probability that the outcome is $0$ is $\frac{1}{2}$. 
	
\end{theorem}

\section{Experimental details}\label{supp:C}
In this section, we will give further experimental details about our work.

First, we will show the results for multi-layer ans\"{a}tze, the probability under all possible observables under small-scale qubit systems, sampled points under different sizes of the dataset with 1000 qubits, and sampled points with 10000 qubits. Our device can support experiments involving up to 10000 qubits maximally.

In the main text we state that the maximal probability for nontrivial sampling is $\frac{1}{4}$. It is not a mathematical conclusion but derived from various experiments under all possible observables as shown in FIG \ref{fig:prob_all_observables}.
\begin{figure}[H]
    \centering
    \includegraphics[width=0.4\linewidth]{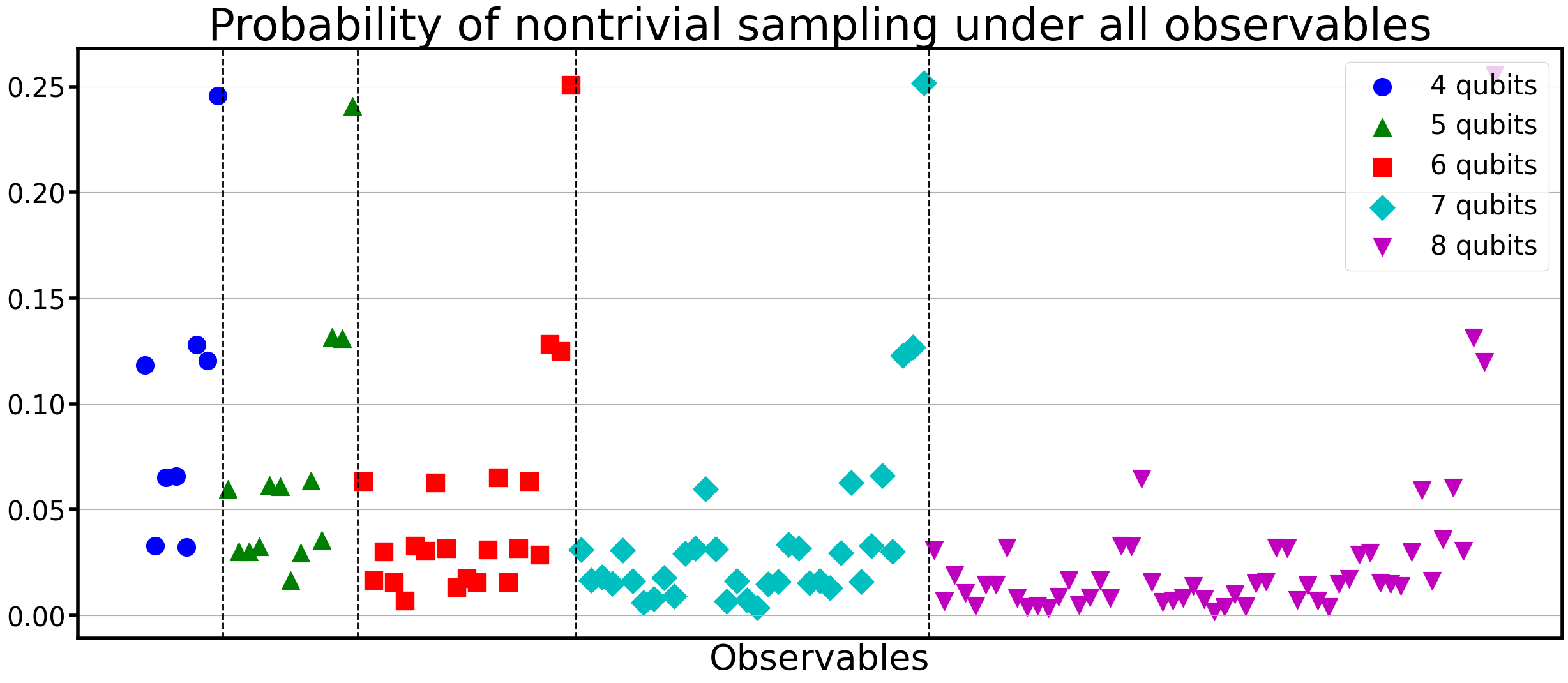}
    \caption{Probability under all possible observables.}
    \label{fig:prob_all_observables},
\end{figure}

In practical applications, a multi-layer ans\"{a}tz with about four layers is commonly employed to enhance the expressive power of quantum circuits, compared to a single-layer ans\"{a}tz. 
Through experiments with multi-layer ans\"{a}tz, we observed that the number of layers does indeed impact sampling efficiency as shown in FIG. \ref{fig:multiansatz}. 
However, while this probability decreases as the number of layers increases, we argue that such a decrease does not hinder the algorithm's scalability to large-scale qubit systems. 
This is because, in practical applications, it is unnecessary to use an excessive number of layers, such as 100 or even 1000 layers. In typical QML or VQA applications, the decrease in sampling efficiency is always limited within circuits of 3 to 10 layers, decreasing from approximately 0.25 down to around 0.13. Such a probability is entirely acceptable for our tasks and does not significantly compromise its effectiveness.
\begin{figure}
    \centering
    \includegraphics[width=0.5\linewidth]{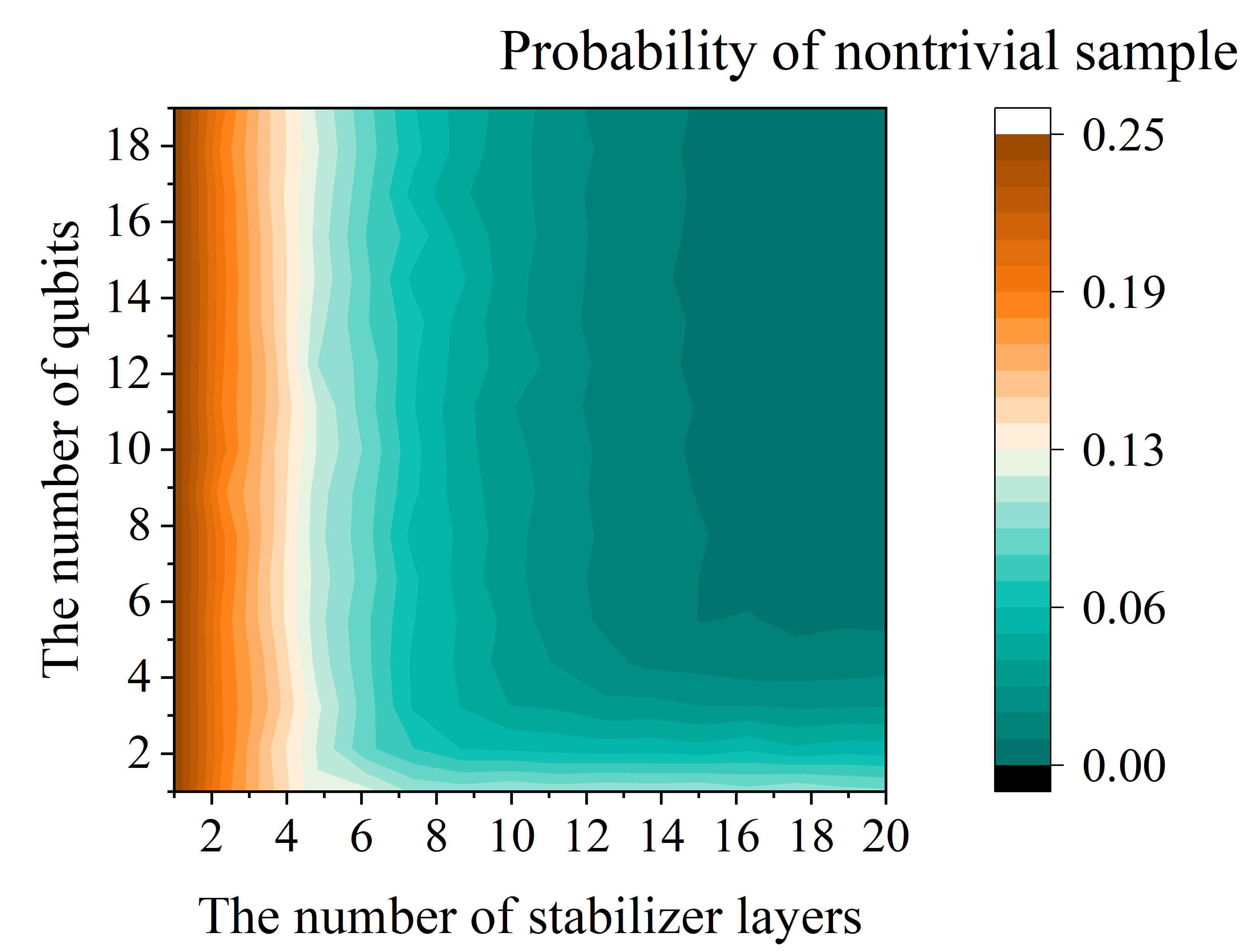}
    \caption{Probability of nontrivial sampling with multi-layer ans\"{a}tze.}
    \label{fig:multiansatz}
\end{figure}

Moreover, in the main text, we provided only a single figure illustrating the relationship between loss, variance, and dataset size. Here, we will present a detailed analysis of the sampling results across various datasets in FIG. \ref{fig:1kqubits_all}. All experiments are conducted with 3 layers ans\"{a}tze.
\begin{figure}[htbp]
    \centering
    \begin{subfigure}[b]{0.32\textwidth}
        \centering
        \includegraphics[width=0.9\textwidth]{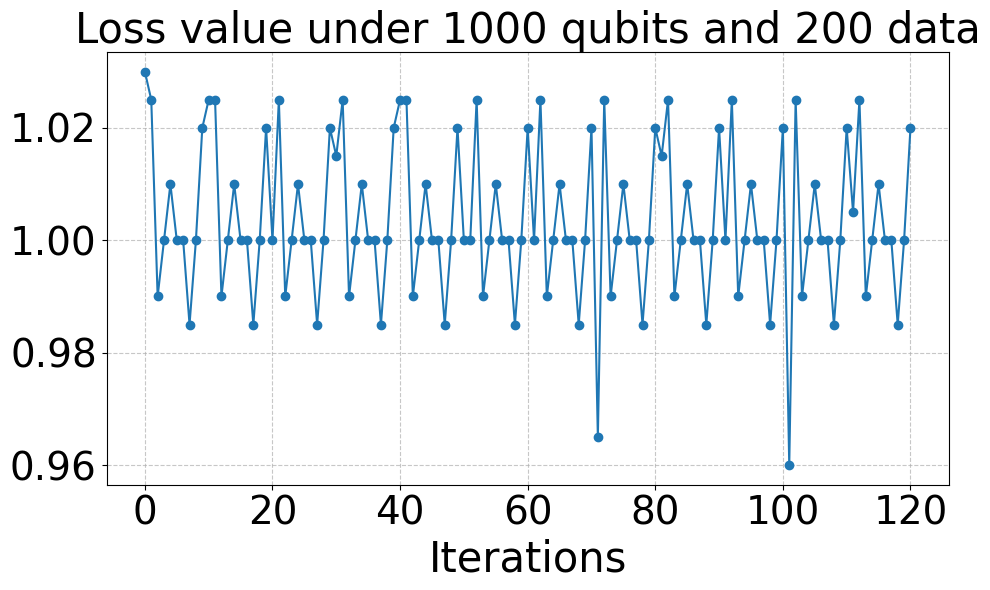}
        \caption{1000 qubits, 200 data.}
        \label{fig:1k200}
    \end{subfigure}
    \begin{subfigure}[b]{0.32\textwidth}
        \centering
        \includegraphics[width=0.9\textwidth]{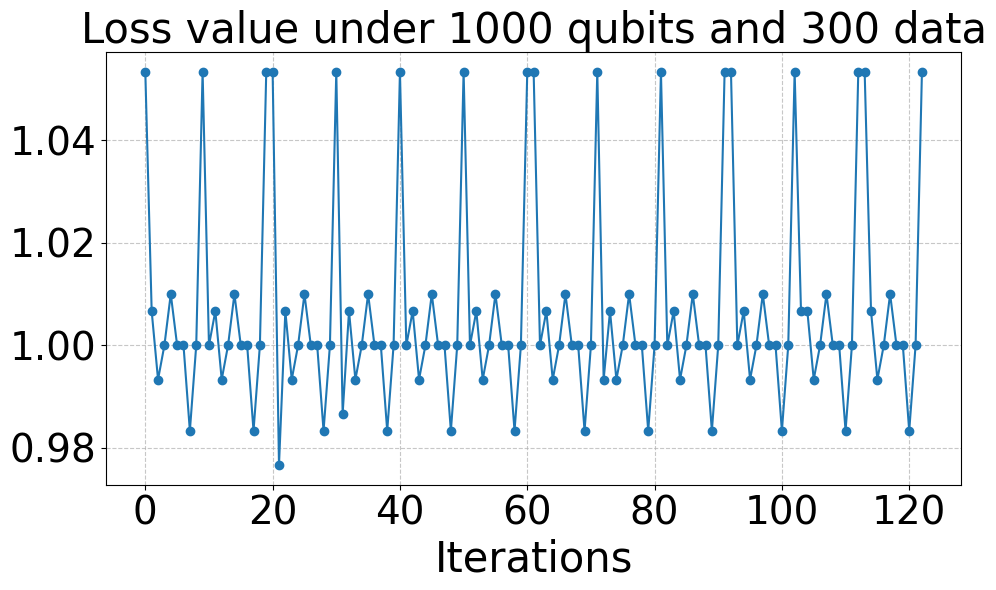}
        \caption{1000 qubits, 300 data.}
        \label{fig:1k300}
    \end{subfigure}
    \begin{subfigure}[b]{0.32\textwidth}
        \centering
        \includegraphics[width=0.9\textwidth]{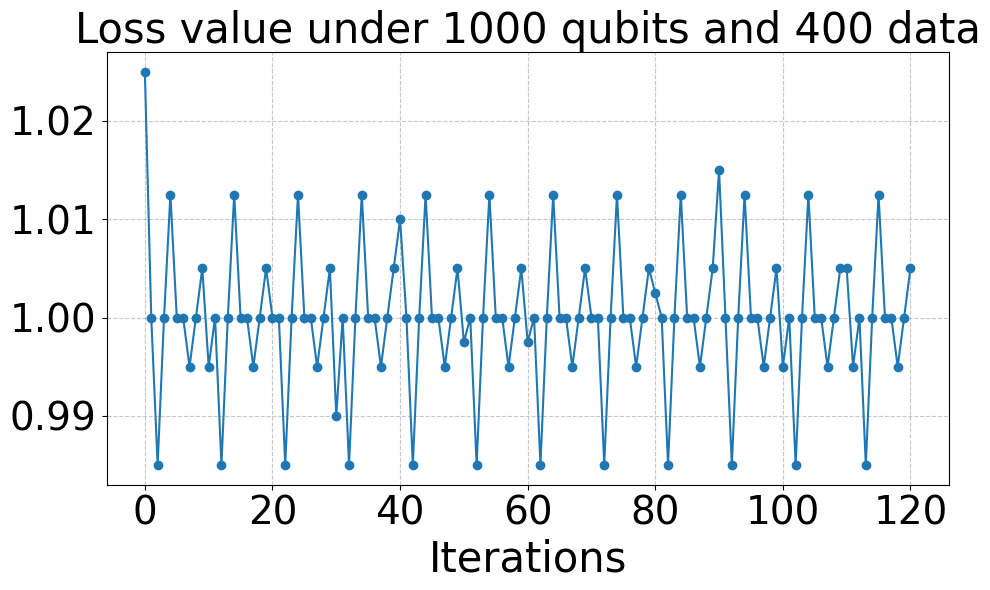}
        \caption{1000 qubits, 400 data}
        \label{fig:1k400}
    \end{subfigure}

    \begin{subfigure}[b]{0.32\textwidth}
        \centering
        \includegraphics[width=0.9\textwidth]{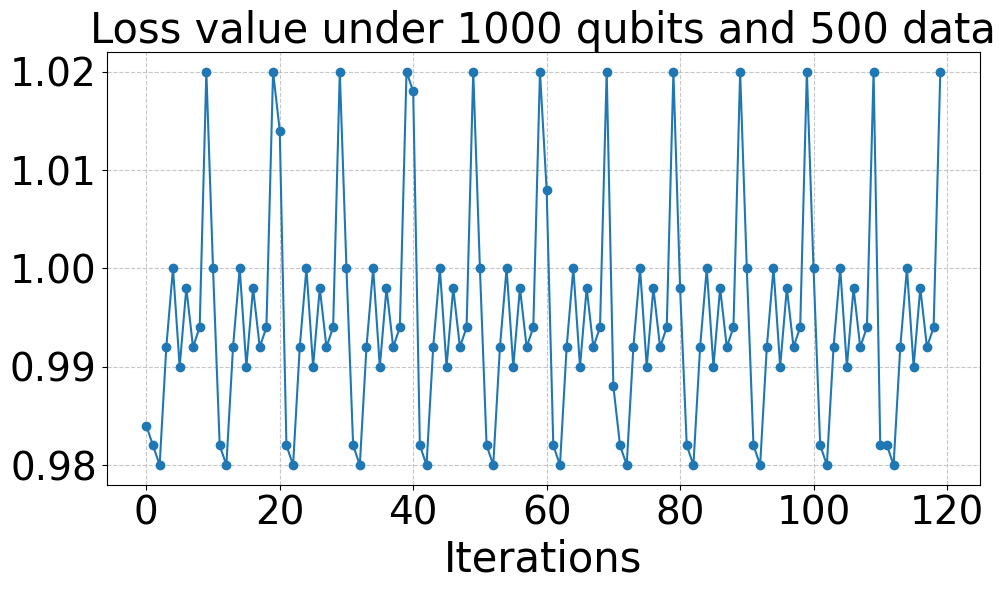}
        \caption{1000 qubits, 500 data.}
        \label{fig:1k500}
    \end{subfigure}
    \begin{subfigure}[b]{0.32\textwidth}
        \centering
        \includegraphics[width=0.9\textwidth]{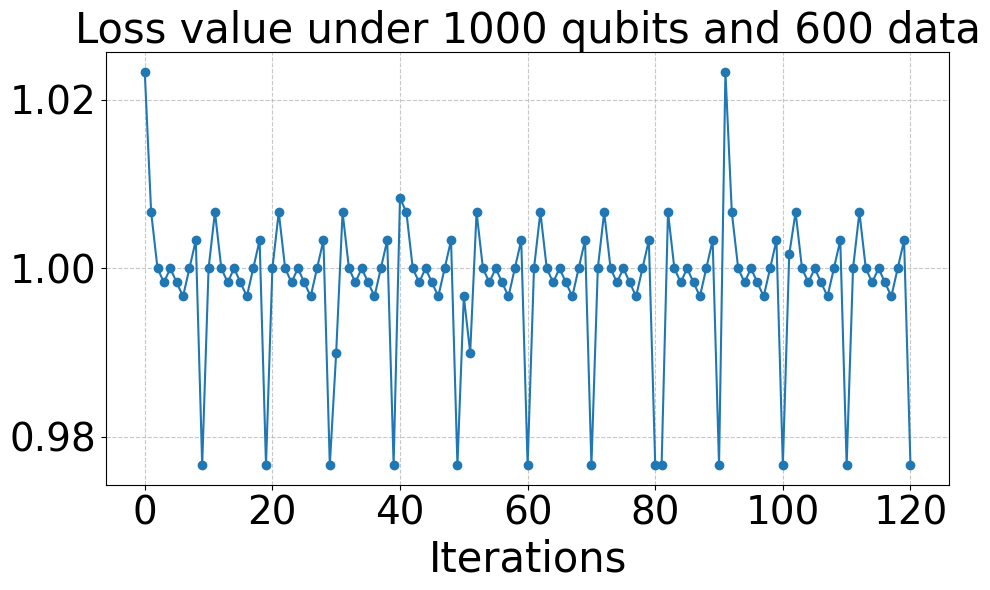}
        \caption{1000 qubits, 600 data.}
        \label{fig:1k600}
    \end{subfigure}
    \begin{subfigure}[b]{0.32\textwidth}
        \centering
        \includegraphics[width=0.9\textwidth]{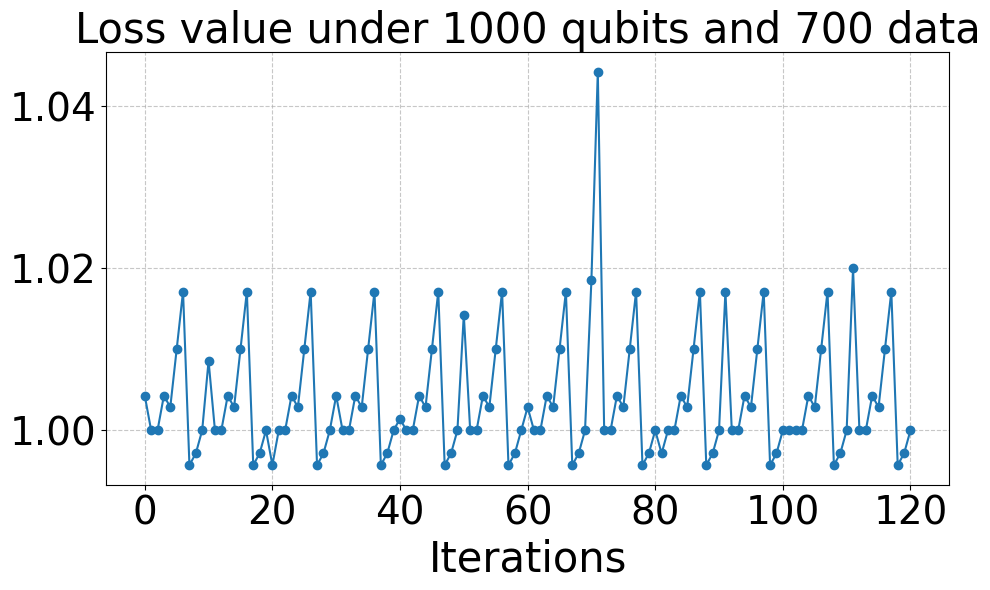}
        \caption{1000 qubits, 700 data.}
        \label{fig:1k700}
    \end{subfigure}

    \begin{subfigure}[b]{0.32\textwidth}
        \centering
        \includegraphics[width=0.9\textwidth]{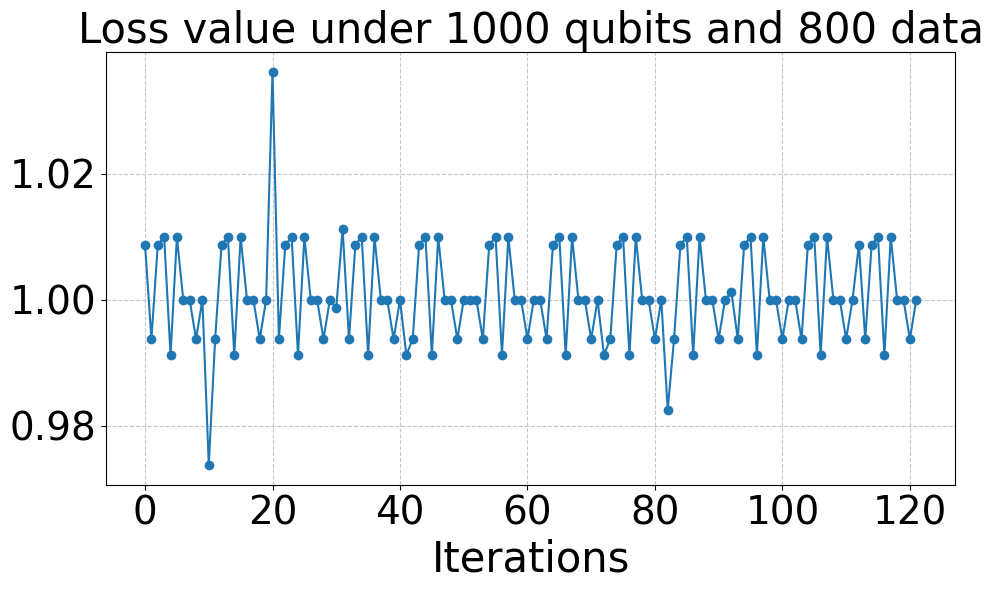}
        \caption{1000 qubits, 800 data.}
        \label{fig:1k800}
    \end{subfigure}
    \begin{subfigure}[b]{0.32\textwidth}
        \centering
        \includegraphics[width=0.9\textwidth]{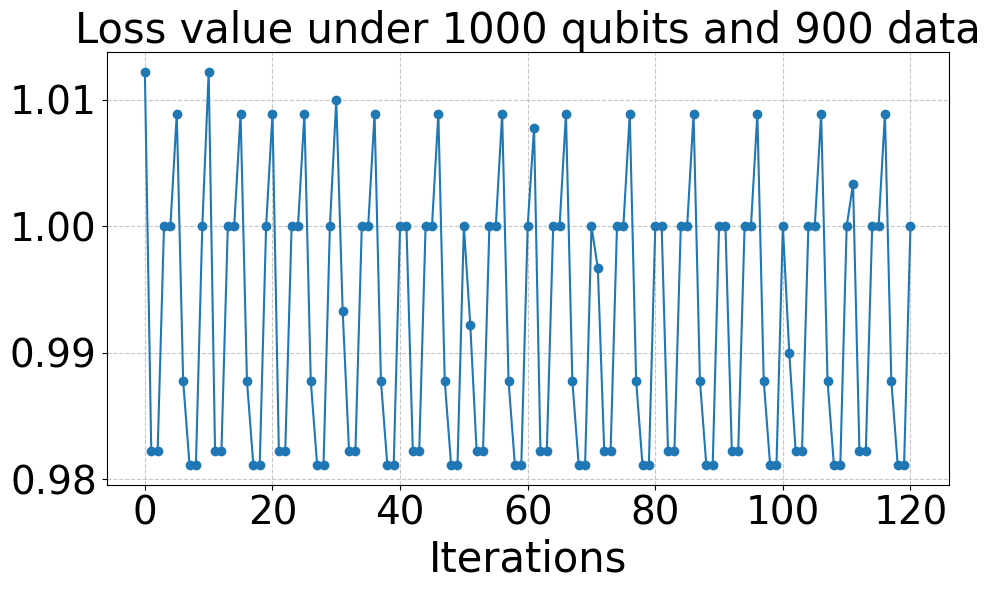}
        \caption{1000 qubits, 900 data.}
        \label{fig:1k900}
    \end{subfigure}
    \begin{subfigure}[b]{0.32\textwidth}
        \centering
        \includegraphics[width=0.9\textwidth]{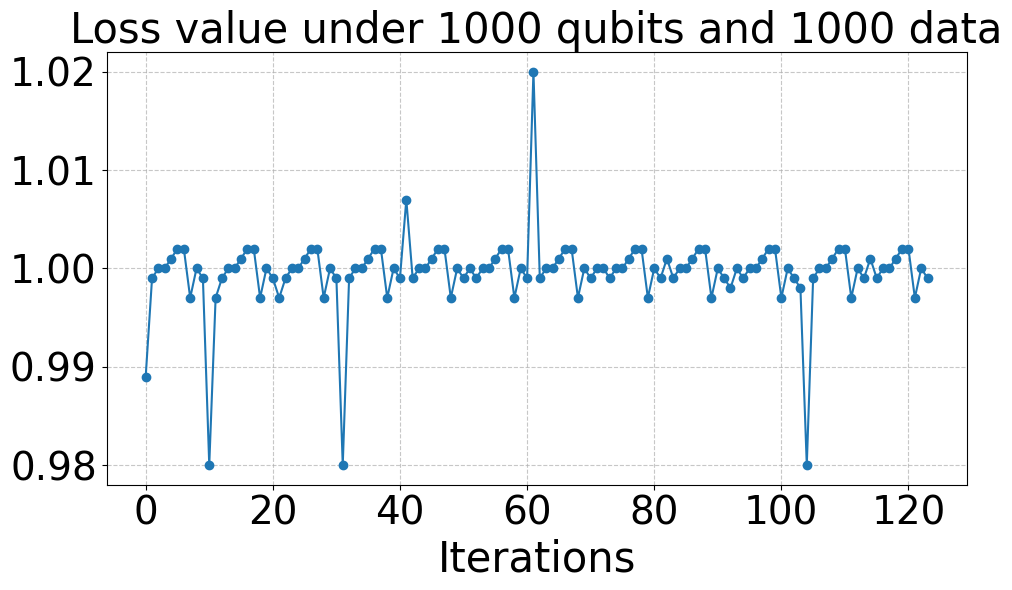}
        \caption{1000 qubits, 1000 data.}
        \label{fig:1k1000}
    \end{subfigure}

    \caption{Results for datasets of varying sizes.}
    \label{fig:1kqubits_all}
\end{figure}

\end{document}